\documentclass[a4,
10pt,
]{
extarticle%
}

\usepackage[utf8]{inputenc}

\usepackage[usenames,dvipsnames,svgnames,table]{xcolor}
\usepackage{pgf}
\usepackage{tikz}
\usepackage{tikz-cd}
\usetikzlibrary{calc}
\tikzset{input/.style={}}
\usetikzlibrary{shapes,arrows}
\usepackage{pgfkeys}
\usepackage{pgfplots}
\usepackage{verbatim}
\usepackage{graphicx}
\usepackage{longtable}
\usepackage{paralist}
\usepackage[]{graphicx,float,latexsym,amssymb,amsmath,
amsfonts,amstext,%
amsthm,%
times}
\usepackage{varioref}
\usepackage{hyperref}
\hypersetup{colorlinks=true,linkcolor=blue,urlcolor=RawSienna,citecolor=RawSienna}
\usepackage[
capitalize,nameinlink]{cleveref}

\newcommand\numberthis{\refstepcounter{equation}\tag{\theequation}}

\usepackage{wrapfig}

\usepackage{palatino} 
\usepackage{mathpazo} 
\usepackage[scaled]{helvet} 
\usepackage{courier} 
\normalfont
\usepackage[T1]{fontenc}

\restylefloat{figure}

\newcommand{\ida}{\stackrel{{\sf def}}=}
\newcommand{\set}[1]{\mbox{$\{#1\}$}}
\newcommand{\pair}[1]{\mbox{$\langle #1\rangle$}}

\newcommand{\union}{\cup}

\newcommand{\intersection}{\cap}
\newcommand{\id}[1]{\mbox{\it #1}}
\newcommand{\N}{\mathbb{N_0}}

\newcommand{\Bool}{\mathbb{B}}
\newcommand{\Rp}{{\mathbb{R}}^{>0}}
\newcommand{\Rnn}{{\mathbb{R}}^{\geq 0}}

\theoremstyle{definition}
\newtheorem{notation}{Notation}[section]

\newtheorem{theorem}{Theorem}[section]
\newtheorem{lemma}{Lemma}[section]
\newtheorem{corollary}{Corollary}[section]
\newtheorem{example}{Example}[section]
\newcommand{\be}{\begin{enumerate}}
\newcommand{\ee}{\end{enumerate}}
\newcommand{\bdes}{\begin{description}}
\newcommand{\edes}{\end{description}}
\newcommand{\bt}{\begin{theorem}}
\newcommand{\et}{\end{theorem}}
\newcommand{\bl}{\begin{lemma}}
\newcommand{\el}{\end{lemma}}
\newtheorem{prop}[theorem]{Proposition}
\newcommand{\bp}{\begin{prop}}
\newcommand{\ep}{\end{prop}}
\newtheorem{defn}[theorem]{Definition}
\newcommand{\bd}{\begin{defn}}
\newcommand{\ed}{\end{defn}}
\newtheorem{remarks}[theorem]{Remarks}
\newcommand{\brem}{\begin{remark}}
\newcommand{\erem}{\end{remark}}
\newcommand{\brems}{\begin{remarks}}
\newcommand{\erems}{\end{remarks}}

\usepackage{stackrel}
\newcommand{\goesto}[2]{\stackrel[#1]{#2}{\longrightarrow}}
\newcommand{\lit}[1]{{\sf{#1}}}
\renewcommand{\t}{\lit{t}}
\newcommand{\h}{\lit{h}}
\newcommand{\e}{\lit{e}}
\newcommand{\tp}{\lit{top}}

\newcommand{\pass}{\lit{pass}}
\newcommand{\override}{!}
\newcommand{\overzero}{\lit{!0}}
\newcommand{\overone}{\lit{!1}}

\newcommand{\true}{\lit{true}}
\newcommand{\false}{\lit{false}}

\newcommand{\Bbot}{B_{\bot}}

\newcommand{\bbot}{b_{\bot}}

\newcommand{\Act}{\id{Act}}
\newcommand{\Cmd}{\id{Cmd}}
\newcommand{\oa}{\overline{a}}
\newcommand{\ob}{\overline{b}}
\newcommand{\oc}{\overline{c}}
\newcommand{\om}{\overline{m}}

\newcommand{\ot}{\overline{\t}}

\newcommand{\opass}{\overline{\pass}}

\newcommand{\obb}{\overline{\bbot}}
\newcommand{\obot}{\overline{\bot}}

\newcommand{\B}{{\cal{B}}}
\newcommand{\I}{{\cal{I}}}

\newcommand{\IR}{\I_{R}}

\renewcommand\numberthis{\refstepcounter{equation}\tag{\theequation}}

\usepackage{wrapfig}
\usepackage{palatino} 
\usepackage{mathpazo} 
\usepackage[scaled]{helvet} 
\usepackage{courier} 
\normalfont
\usepackage[T1]{fontenc}
\usepackage{varioref}
\usepackage{color}
\usepackage{hyperref}
\hypersetup{colorlinks=true,linkcolor=blue,urlcolor=RawSienna,citecolor=RawSienna
}
\usepackage[
capitalize,nameinlink]{cleveref}

\title{Generalised Dining Philosophers as Feedback Control}

\author{\Author{Venkatesh Choppella}\\ 
    \Address{IIIT Hyderabad, India}\\
     \Email{venkatesh.choppella@iiit.ac.in}\\
    \and
       \Author{Kasturi Viswanath}\\
        \Address{IIIT Hyderabad, India}
        \Email{viswanath.iiithyd@gmail.com}
        \Author{Arjun Sanjeev}\\
        \Address{IIIT Hyderabad, India}
        \Email{arjun.sanjeev@research.iiit.ac.in}
      }

\author{Venkatesh Choppella, Kasturi Viswanath and Arjun Sanjeev}

\date{}

\begin{document}

\maketitle
\tableofcontents

\begin{abstract}
We revisit the Generalised Dining Philosophers problem
through the perspective of feedback control.  The result is
a modular development of the solution using the notions of
system and system composition (the latter due to Tabuada) in
a formal setting that employs simple equational reasoning.
The modular approach separates the solution architecture
from the algorithmic minutiae and has the benefit of
simplifying the design and correctness proofs.

Three variants of the problem are considered: N=1, and N > 1
with centralised and distributed topology.  The base case
(N=1) reveals important insights into the problem
specification and the architecture of the solution.  In each
case, solving the Generalised Dining Philosophers reduces to designing
an appropriate feedback controller.
\end{abstract}

\section{Introduction: The Dining Philosophers problem}
\label{sec:intro}

Resource sharing amongst concurrent, distributed processes
is at the heart of many computer science problems, specially
in operating, distributed embedded systems and networks.
Correct sharing of resources amongst processes must not only
ensure that a single, non-sharable resource is guaranteed to
be available to only one process at a time (safety), but
also starvation-freedom -- a process waiting for a resource
should not have to wait forever.  Other complexity metrics
of interest in a solution are average or worst case waiting
time, throughput, etc.  Distributed settings introduce other
concerns: synchronization, faults, etc.

The Dining Philosophers problem, originally formulated by
Edsger Dijkstra in 1965 and subsequently published in
1971\cite{dijkstra-1971} is a celebrated thought experiment
in concurrency control: Five philosophers are seated around
a table.  Adjacent philosophers share a fork.  A philosopher
may either eat or think, but can also get hungry in which
case he/she needs the two forks on either side to eat.
Clearly, this means adjacent philosophers do not eat
simultaneously, the \emph{exclusion} condition.  Each
philosopher denotes a process running continuously and
forever that is either waiting (\emph{hungry}) for a
resource (like a file to write to), or using that resource
(\emph{eating}), or having relinquished the resource
(\emph{thinking}).  The problem consists of designing a
protocol by which no philosopher remains hungry
indefinitely, the \emph{starvation-freeness} condition,
assuming each eating session lasts only for a finite
time\footnote{Other significant works on the Dining
  Philosophers problem
  \cite{Chandy:1984:DPP:1780.1804,chandy-misra-book-1988}
  call this the `fairness' condition.  We avoid this
  terminology since in automata theory `fairness' has a
  different connotation.}.  In addition, the \emph{progress}
condition means that there should be no deadlock: at any
given time, at least one philosopher that is hungry should
move to eating after a bounded period of time.  Note that
starvation-freedom implies progress.

The generalisation of the Dining Philosophers involves N
philosophers with an arbitrary, non-reflexive neighbourhood.
Neighbours can not be eating simultaneously.  The
generalised problem was suggested and solved by Dijsktra
himself\cite{dijkstra-ewd625}.  The Generalised Dining
Philosophers problem is discussed at length in Chandy and
Misra's book\cite{chandy-misra-book-1988}.  The Dining
Philosophers problem and its generalisation have spawned
several variants and numerous solutions throughout its long
history and is now staple material in many operating system
textbooks.

\textbf{Individual Philosopher dynamics:} Consider a
single philosopher who may be in one of three states:
thinking, hungry and eating.  At each step, the
philosopher may \emph{choose} to either continue to be in
that state, or switch to the next state (from thinking to
hungry, from hungry to eating, or from eating to thinking
again).  The dynamics of the single philosopher is shown
in \cref{fig:dp-states}.   Note that the philosophers run
forever.

\usetikzlibrary{arrows,automata}

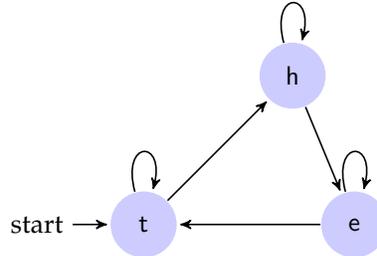
\begin{figure}
\caption{Philosopher states and transitions\label{fig:dp-states}}
\centering{
\begin{tikzpicture}[->,>=stealth',shorten >=1pt,auto,node distance=2.8cm,
                    semithick]
  \tikzstyle{every state}=[fill=blue!20,draw=none,text=black]

  \node[initial,state] (T)                    {$\t$};
  \node[state]         (H) [above right of=T] {$\h$};
  \node[state]         (E) [right of=T]       {$\e$};

  \path (T) edge [loop above] node {} (T)
            edge              node {} (H)
        (H) edge [loop above] node {} (H)
            edge              node {} (E)
        (E) edge [loop above] node {} (E)
            edge              node {} (T);
\end{tikzpicture}
}
\end{figure}

We now consider the Generalised Dining Philosophers
problem.
\begin{defn}[\textbf{Generalised Dining Philosophers problem}]
\label{defn:gdp}

N philosophers are arranged in a connected \emph{conflict}
graph $G=\pair{V,E}$ where $V$ is a set of $N=|V|$
philosophers and $E$ is an irreflexive adjacency relation
between them.

   If each of the N philosophers was to continue to evolve
   according to the dynamics in \cref{fig:dp-states}, two
   philosophers sharing an edge in $E$ could be eating
   together, violating \emph{safety}.   The Generalised Dining
   Philosophers problem is the following:
 
  \textbf{Problem}: Assuming that no philosopher eats
  forever in a single stretch, construct a protocol that
  ensures

   \begin{enumerate} 

    \item \textbf{Safety:} No two adjacent philosophers eat
      at the same time.   

    \item \textbf{Starvation-freedom:} A philosopher that is hungry
      eventually gets to eat.

    \item \textbf{Non-intrusiveness:} A philosopher that is
      thinking or eating continues to behave according to
      the dynamics of \cref{fig:dp-states}.
   \end{enumerate}
\end{defn}

A ``protocol'' is usually interpreted as an algorithm or a
computer program that runs as part of the process.  The
protocol defines the interaction between the actors
mentioned (philosophers in the generalised problem, and in
the 5-diners problem, the forks as well) in the problem that
may be needed to solve the problem.

The main actors in this problem are the philosophers.  The
dynamics of the philosophers' transitions as shown in
\cref{fig:dp-states} are governed by their \emph{choice} to
either remain in the same state or move to a new state.
This choice manifests as non-determinism in the dynamics.
The first observation is that the philosopher dynamics as
shown in \cref{fig:dp-states} is inadequate to ensure
safety.  As noted above, nothing prevents two adjacent and
hungry philosophers to both move to eating.

One way of solving the Generalised Dining Philosophers problem is to
define a more complex dynamics that each of the N
philosophers implements so that the safety and starvation
freedom conditions hold.  Yet another way, that hints at the
control approach, is to consider additional actors that
restrain the philosophers' original actions in some specific and
well-defined way so as to achieve safety and starvation
freedom.  The additional actors needed to restrict the
philosophers' actions are called \emph{controllers}.  The
role of the controller is to issue \emph{commands} that may
involve overriding the philosopher's own choice to move to a
new state.  For example, a hungry philosopher who wishes to
eat in the next cycle may find his/her wish overridden by a
command issued by the controller to continue to remain
hungry in the interest of preserving the safety invariant.
However, in any solution to the problem, the controller
should eventually promote a hungry philosopher to eating so
as to preserve the starvation freedom invariant.  It is this
approach that we wish to explore in this paper.

Any control on the philosophers should be not overly
restrictive: a philosopher who is either thinking or eating
should be allowed to exercise his/her choice about what to
do next; only a hungry philosopher may be commanded by the
controller either to continue to remain hungry or switch to
eating, overriding the philosopher's own choice of whether
to stay hungry or switch to eating\footnote{Sometimes, we
  may want to relax this condition: a \emph{preemptive}
  controller may force an eating philosopher back to a
  hungry state if the philosopher eats for too long.
  Preemptive controller design is not discussed, but may be
  implemented using the same ideas as discussed in this
  paper.}.

Solutions to the Generalised Dining Philosophers may be
broadly classified as either centralised or distributed.
The centralised approach assumes a central controller that
commands the philosophers on what to do next.  The
distributed approach assumes no such centralised authority;
the philosophers are allowed to communicate to arrive at a
consensus on what each can do next.

The objective of this paper is to formulate the Generalised
Dining Philosophers problem using the idea of control,
particularly that which involves feedback.  Feedback
control, also called supervisory control, is the foundation
of much of engineering science and is routinely employed in
the design of embedded systems.  However, its value as a
software architectural design principle is only
insufficiently captured by the popular
``model-view-controller'' (MVC) design
pattern\cite{Gamma-et-al-1994}, usually found in the design
of user and web interfaces.  For example, MVC controllers
implement open loop instead of feedback control.

The starting point is a more precise statement of the
problem by employing the formalism of discrete state
transition systems with output, also called Moore machines.
We then borrow the notion of system composition due to
Tabuada\cite{tabuada-book-2009}.  Composition is defined
with respect to an interconnect that relates the states and
inputs of the two systems being composed.  Viewed from this
perspective, the Generalised Dining Philosophers form a system
consisting of interconnected subsystems.  A special case of
the interconnect which relates inputs and outputs yields
modular composition and allows the Generalised Dining Philosophers to be
treated as an instance of feedback control.  The solution
then reduces to designing two types of components - the
philosophers and the controllers - and their interconnections
(the system architecture), followed by definitions of the
transition functions of the philosophers and the
controllers.  The transition function of the controller is
called a \emph{control law}.

The compositional approach encourages us to think of the
system in a modular way, emphasising the interfaces between
components and their interconnections.  One benefit of this
approach is that it allows us to define multiple types of
controllers (N=1, N>1 centralised and N>1 local) that
interface with a fixed philosopher system.  The modularity
in architecture also leads to modular correctness proofs of
safety and starvation freedom.  For example, the proof of
the distributed case is reduced to showing that the
centralised controller state is reconstructed by the union
of the states of the distributed local controllers.  That
said, however, subtle issues arise even in the simplest
variants of the problem.  These have to do with
non-determinism, timing and feedback, but equally, from
trying to seek a precise definition of the problem
itself\footnote{``In the design of reactive systems it is
  sometimes not clear what is given and what the designer is
  expected to produce.''  Chandy and
  Misra\cite[p.~290]{chandy-misra-book-1988}.}.

\paragraph{Paper roadmap}
The rest of the paper is an account on how to solve the
Generalised Dining Philosophers problem in a step-by-step manner,
varying both the complexity of the problem from N=1 to N>1
and from centralised to distributed.  We begin with a short
review of the fundamental concept of systems, their
behaviour and composition (\cref{sec:systems}) and the role
of time.  We then turn our attention to the simplest variant
of the Generalised Dining Philosophers: the 1 Diner problem
(\cref{sec:1dp}) and explore a series of architectures for
the single philosopher, the controller and their
composition.  The architecture identifies the boundaries and
interfaces of each subsystem and the ``wiring'' of the
subsystems, in this case, the one philosopher and the
controller, with each other.  Next we consider N Diners
(\cref{sec:ndp}), where the problem reduces to designing a
controller and (a) defining a data structure internal to the
controller's state and an algorithm to manipulate it and,
(b) computing the set of control inputs.  After proving the
correctness of the centralised solution, we consider
distributed control (\cref{sec:dist}).  The problem here
reduces to distributing the effort of the centralised
controller to N different local controllers, each
controlling the behaviour of its corresponding philosopher.
A clear interconnect boundary between each component defines
exactly which part of the state is shared between the
components.  We compare the feedback control based solution
with other approaches (\cref{sec:related}) and conclude with
some pointers to future work (\cref{sec:conc}).

No prior background in control theory is assumed; relevant
concepts from control systems are explained in the next
section.

\section{Systems Approach}
\label{sec:systems}
The main idea in control theory is that of a \emph{system}.
Systems have \emph{state} and exhibit behaviour governed by
a \emph{dynamics}.  A system's state undergoes change due to
input.  Dynamics is the unfolding of state over time,
governed by laws that relate input with the state.  A
system's dynamics is thus expressed as a relation between
the current state, the current input and the next state.
Its output is a function of the state.  Thus inputs and
outputs are connected via state.  The system's state is
usually considered hidden, and is inaccessible directly.
The observable behaviour of a system is available only via
its output.   A schematic diagram of a system is shown in
\Cref{fig:auto-system}.

\usetikzlibrary{shapes,arrows}

\tikzstyle{block} = [draw, fill=blue!20, rectangle, 
    minimum height=3em, minimum width=6em]
\tikzstyle{sum} = [draw, fill=blue!20, circle, node distance=1cm]
\tikzstyle{input} = [coordinate]
\tikzstyle{output} = [coordinate]
\tikzstyle{pinstyle} = [pin edge={to-,thin,black}]

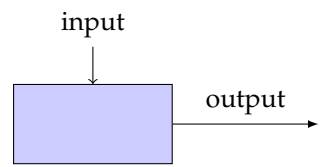
\begin{wrapfigure}{r}{6cm}
\caption{System with input and output\label{fig:auto-system}}
\centering{
\begin{tikzpicture}[auto, node distance=2cm,>=latex]
    \node [block,%
           pin={[pinstyle]above:input},
            node distance=3cm] (system) {};

    \node [output, right of=system, node distance=3cm]  (output) {};
    \draw [->] (system) -- node[name=o] {output} (output);

\end{tikzpicture}
} 
\end{wrapfigure}

In control systems, we are given a system, often identified
as the \emph{plant}.  The plant is also referred to as
\emph{model} in the literature and we shall use the two
terms interchangeably.  The plant exhibits a certain
observable \emph{behaviour}.  The behaviour may be
informally described as an infinite sequence of input-output
pairs.  In addition to the plant, we are also given a
\emph{target behaviour} that is usually a restriction of the
plant's behaviour.

There are many ways of realising the target behaviour.  The
first is to attach another system, an \emph{output filter},
that takes the output of the plant and suitably modifies it,
so that the resulting output now conforms to the behaviour
specified in the problem.  A second way is to have a system,
\emph{input filter}, that intercepts the inputs to the
plant, suitably modifies (or restricts) them and then feeds
the results to the plant.  In both cases, the architecture
of the plant is left untouched.  The dynamics of the plant
too remains unaffected.  Restricting the input is, however,
not always possible.

There is a third way to influence the plant to achieve the
specified behaviour, which is usually what is referred to as
\emph{control}.  The control problem is, very roughly, the
following: \emph{what additional} input(s) should be
supplied to the plant, such that the resulting dynamics as
determined by a new relation between states and inputs now
exhibits output behaviour that is either equal or
approximately equal to the target behaviour specified in the
problem?  The additional input is usually called the
\emph{forced} or \emph{control input}.  Notice that the
additional inputs may require altering the interface and the
dynamics of the plant.  The plant's altered dynamics need to
take into account the combined effect of the original input and the
control input.

The second design question is \emph{how} should the control
input be computed.  Often the control input is computed as a
function of the output of the plant (now extended with the
control input).  Thus we have another system, the
\emph{controller}, (one of) whose inputs is the output of
the plant and whose output is (one of) the inputs to the
plant.  This architecture is called \emph{feedback control}.
The relation between the controller's input and its state
and output is called a \emph{control law}. 


\Cref{fig:feedback-system} is a schematic diagram
representing a system with feedback control.

\usetikzlibrary{shapes,arrows}

\tikzstyle{block} = [draw, fill=blue!20, rectangle, 
    minimum height=3em, minimum width=6em]
\tikzstyle{sum} = [draw, fill=blue!20, circle, node distance=1cm]
\tikzstyle{input} = [coordinate]
\tikzstyle{output} = [coordinate]
\tikzstyle{point} = [coordinate]
\tikzstyle{pinstyle} = [pin edge={to-,thick,black}]

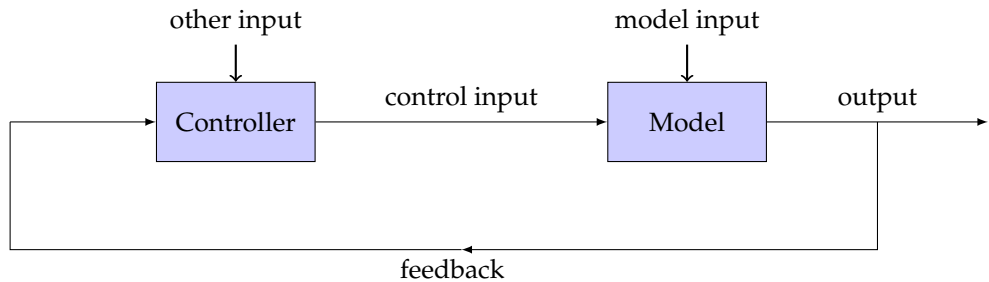
\begin{figure}
\caption{Feedback Control System\label{fig:feedback-system}}
\centering{
\begin{tikzpicture}[auto, node distance=2cm,>=latex]
    \node [input, name=input] {};
    \node [block, right of=input, node distance=3cm, pin={[pinstyle]above:other input}] (controller) {Controller};
    \node [block, right of=controller, pin={[pinstyle]above:model input},
            node distance=6cm] (model) {Model};
    \draw [->] (controller) -- node[name=u] {control input} (model);
    \node [output, right of=model, node distance=4cm] (output) {};
    \node [point, below of=u] (feedback) {feedback};

    \draw [draw,->] (input) -- node {} (controller);
    \draw [->] (model) -- node [name=y] {output}(output);
    \draw [->] (y) |- node[name=f] {} (feedback);
    \draw [-] (feedback) -| node[pos=0.01] {feedback}  node [near end] {} (input);
\end{tikzpicture}
} 
\end{figure}

The principle of feedback control is well studied and is
used extensively in building large-scale engineering systems
of wide variety.  A modern introduction to the subject is
the textbook by {\AA}str{\"o}m and
Murray\cite{astrom-murray-book-2008} which motivates the
subject by illustrating the use of feedback control in
various engineering and scientific domains: electrical,
mechanical, chemical, and biological and also computing.


\label{subsec:formal-notion-system}

In the rest of this section, we present the formal notion of a
system and system composition as defined by
Tabuada\cite{tabuada-book-2009}.

\subsection{Specification, instance and behaviour}
\label{subsec:system-spec}

A \emph{system specification} (or \emph{type}) is a tuple
with six fields:

\[S = \pair{x:X, x^{0}:X^{0}, u:U, \goesto{S}{}, y:Y, h}\]

\noindent $X$, a set of \emph{states} is called the
\emph{state space} of the system $S$.  $X^{0}\subseteq X$ is
the set of \emph{initial states}.  $U$ is the \emph{input}
space.  $\goesto{}{}\subseteq X\times\ U\times X$ is the
\emph{input-state transition relation} describing the set of
possible transitions, that is the dynamics.  $Y$ is the
\emph{output} space, and $h:X\longrightarrow Y$ is the
output function that maps states to outputs.  We elide the
component $X^{0}$ if $X^{0}=X$.  $U$ is \emph{unitary} if it
is a singleton $\set{*}$.  We elide $U$, $Y$ or $h$ if $U$
is unitary, $Y=X$ or $h$ is identity, respectively.

\begin{notation}
\label{not:sys-field}

$S$ is to be seen as a record with canonical field
\emph{names} $X$, $U$, $X^{0}$, $\goesto{S}{}$, $Y$ and $h$.
These names are associated with values when a system is
defined.  When a system is defined, we indicate the values
in place, say if $X$ equals $A$, we write
$S=\pair{X=A,\ldots}$.  The lowercase names $x$,$x^{0}$,$u$
and $y$ denote \emph{system variables} that range over $X$,
$X^{0}$, $U$ and $Y$ respectively.  Their values define the
configuration of the system at any point in its evolution:
the value of the state, the initial state, input, and
output.  A field or system variable like $X$ or $x$ of a
system $S$ is written $X_S$, alternatively $S.X$.  When $S$
is clear from the context, the subscript is omitted.  Often,
additional variable names (aliases) are used.  E.g., in the
philosopher system $Q$ defined later, the name $x$ denoting
the state system variable is aliased to $a$.  We write
$S=\pair{a:X=A,\ldots,}$ to denote that the state space
field $X$ has the value $A$, and the variable name $a$ is an
alias to the (default) state variable $x$.
\end{notation}

$S$ is \emph{deterministic} if the relation $\goesto{S}{}$
is a partial function, \emph{non-deterministic} otherwise.
The transition relation in a deterministic system is usually
denoted by a transition \emph{function} $f:X\times
U\rightarrow X$.  A system is \emph{autonomous} if $U_S$ is
unitary, non-autonomous otherwise.  A system is
\emph{transparent}, or \emph{white box} if its output
function $h$ is identity.

A \emph{system instance $s$ of type $S$}, denoted $s:S$, is
a record consisting of three fields $\pair{x:X_S, u:U_S,
  y:Y_S}$.  The fields of $s$ are accessed via the dot
notation. E.g., $s.x$, etc.  We also write $s.X$ to mean
$S.X$ where $s:S$, etc.  Often times, we overload a system
specification $S$ to also denote its instance.  Thus $x_S$
denotes the state of the system instance of type $S$, etc.

\subsection{System composition}

A complex system is best described as a composition of
interconnected subsystems.  We employ the key idea of an
interconnect due to Tabuada\cite{tabuada-book-2009}.  An
\emph{interconnect} between two systems is a relation that
relates the states and the inputs of two systems.

Let $S_c=\pair{X_c, X_c^0, U_c, \goesto{c}{}, Y_c, h_c}$ and
$S_a=\pair{X_a, X_a^0, U_a, \goesto{a}{}, Y_a, h_a}$ be two
systems, then $\I\subseteq X_c \times X_a \times U_c\times
U_a$ is called an \emph{interconnect} relation.  Informally,
an interconnect specifies the architecture of the composite
system.

The composition $S_c\times_{\I}S_a$ of $S_c$ and $S_a$ with
respect to the interconnect\ $\I$\ is defined as the system
$S_{ca}=\pair{X_{ca}, X_{ca}^0, U_{ca}, \goesto{ca}{},
  Y_{ca}, h_{ca}}$, where
\begin{enumerate}
\item $X_{ca}=\set{(x_c,x_a)\ |\ \exists
  u_c,u_a. (x_c,x_a,u_c,u_a)\in \I}$
\item $X^0_{ca}=X_{ca}\intersection (X^0_c\times X^0_a)$
\item $U_{ca} = U_c\times U_a$
\item $(x_c,x_a)\goesto{ca}{u_c,u_a}(x'_c , x'_a)$ iff 
  \begin{enumerate}
    \item $x_c\goesto{c}{u_c}x'_c$, 
    \item $x_a\goesto{a}{u_a}x'_a$, and
    \item $(x_c ,x_a ,u_c , u_a)\in \I$
  \end{enumerate}
\item $Y_{ca} = Y_c\times Y_a$
\item $h_{ca}(x_c , x_a) = (h_c(x_c),h_a(x_a))$.
\end{enumerate}

\begin{notation}[Components of a composite]
If $D = C\times_{\I} A$ then $D.C$ and $D.A$ refer to the
projections of the respective subsystems.  The individual
values e.g., $D.C.x$ and $D.A.x$ are abbreviated $D.x_C$ and
$D.x_A$.  When the composition $D$ is clear from the
context, we continue to use $x_C$ and $x_A$, etc.
\end{notation}

The restrictions of systems $A$ and $C$ are embedded as
subsystems in the composite system.  The interconnected
system $A$ is different from the $A$ that is unconnected.
The former's dynamics is governed by the additional
constraints imposed by the interconnect.  Often we will
define a system $A$, and then its composition with another
system.  Subsequent references to $A$ and its behaviour
refer to the interconnected (and hence constrained)
subsystem of the composite system.  This could occasionally
lead to some ambiguity, specially when the interconnect is
not clear from the context.  In such a case, the context
will be made clear.

Second, since the interconnect completely defines the
composite system, we will limit our description of the
composite system to the individual subsystems and the
interconnect relation and rarely write down the components
of the composite systems.

The notion of Tabuada composition subsumes several other
notions of composition.

\begin{example}[Synchronous Composition]
The synchronous
composition~\cite{hoare-csp-book-1985,Miremadi-et-al-2008},
also called ``parallel composition with shared
actions''~\cite{magee-kramer-2006}, is one in which two
systems have input alphabets with possibly non-empty
intersection.  The two systems simultaneously transition on
any input that is in the intersection; otherwise, each
system transition to the input that is in its input space,
the other process does not advance.

The synchronous composition $S_c\times_{H} S_a$ of two white
box systems $S_c$ and $S_a$ is given by
\[S_c\times_{H} S_a = \pair{X_c\times X_a, X^0_c\times X^0_a,
  U_c \union U_a, \goesto{H}{}}\]

\noindent where

\be
\item $(x_c,x_a)\goesto{H}{u}(x'_c,x'_a)$ if $u\in
  U_c\intersection U_a$, $x_c\goesto{c}{u}x'_c$ and
  $x_a\goesto{a}{u}x'_a$ 

\item $(x_c,x_a)\goesto{H}{u}(x'_c,x_a)$ if $u\in U_c\setminus
  U_a$, $x_c\goesto{c}{u}x'_c$

\item $(x_c,x_a)\goesto{H}{u}(x_c,x'_a)$ if $u\in U_a\setminus
  U_c$, $x_a\goesto{a}{u}x'_a$
\ee

This may be expressed as a Tabuada composition over systems
whose input spaces are lifted by a fresh element $\bot$,
distinct from elements in $U_c\union U_a$.

\begin{align*}
S_{c'} = \pair{X_c, X^{0}_c, U_c\union\set{\bot}, \goesto{c'}{}}\\
S_{a'} = \pair{X_a, X^{0}_a, U_a\union\set{\bot}, \goesto{a'}{}}
\end{align*}

\noindent where 

\be
  \item $\goesto{c'}{} = \goesto{c}{} \union
    \set{x\goesto{c'}{\bot}x\ | x\in X_c}$

  \item $\goesto{a'}{} = \goesto{a}{} \union
    \set{x\goesto{a'}{\bot}x\ | x\in X_a}$
\ee

The interconnect $\I_{H}$ is defined as $\I_{H}\subseteq
X_c\times X_a\times (U_c\union{\bot})\times (U_a\union{\bot})$, where
\[(x_c,x_a,u_c,u_a)\in \I_{H} \ \mbox{iff}\]

\noindent any of the following hold:

\begin{align*}
u_a = u_c  & \quad \mbox{and}\quad u_a\in U_a\intersection
U_c,\ \mbox{or}\\
u_a = \bot & \quad \mbox{and}\quad u_c\in U_c\setminus U_a,\ \mbox{or}\\
u_c = \bot & \quad \mbox{and}\quad u_a\in U_a\setminus U_c
\end{align*}

It is a simple exercise to verify that $(x_c,
x_a)\goesto{H}{u}(x'_c, x'_a)$ iff $(x_c,
x_a)\goesto{\I_{H}}{(u_c,u_a)}(x'_c, x'_a)$.
\end{example}

\subsection{Modular interconnects}
While interconnects can, in general, relate states, the
interconnects designed in this paper are \emph{modular}:
they relate inputs and outputs.  Defining a modular
interconnect is akin to specifying a wiring diagram between
two systems.  Modular interconnects drive modular design.

\subsection{Time}
\label{subsec:discrete-time}

The Generalised Dining Philosophers problem refers to time
in the assumption (``forever'') and the safety (at the
``same time'') and starvation freedom requirements (``eventually'').
Therefore, any solution to the Generalised Dining Philosophers problem
will need to be based on a model of time that addresses both
simultaneity and eternity.

It is useful to interpret $x\goesto{}{u}x'$ happening over
time: the system is first in state $x$, and then as a result
of input $u$, it comes into the state $x'$.  Nothing,
however is mentioned about \emph{when} that transition
happens in the formal definition of a system.  It is
therefore, necessary to introduce an explicit notion of time
as part of the system's dynamics.  There are several models
of time, but the modelling of time in systems is
subtle\cite{Lee:2009:CNT:1506409.1506426,lee-seshia-book-2e-2015}.
Models range from ordinal time which denotes the number of
transitions made by the system, to physical time as a
non-negative positive real number.  We assume that the
reference to time in the Generalised Dining Philosophers problem is to
physical time and not ordinal time because the Philosophers
represent processes executing in physical time.  Given this
assumption, the Generalised Dining Philosophers is a good example of the
inadequacy of ordinal time.  Consider two adjacent
philosophers $A$ and $B$ and their trajectories in real time
(\cref{tbl:2phil}).  Each philosopher, after the second step
is in the eating state.  But these steps are disjoint in
physical time: $A$ is eating from 2.0 to 7.0 seconds,
whereas $B$ doesn't start eating till 8.0 seconds.  Although
$A[2] = \e = B[2]$, these states do not coincide in physical
time.  The converse may also happen: $A$ and $B$'s states do
not coincide in ordinal time, but coincide in physical time
($t=12.0s$ onward in \cref{tbl:2phil}).  See
\cref{fig:2phil-unsynchronized}.

\begin{table}[h]
  \caption{Trajectories of states of two philosophers over
    physical time.  Note that ordinal time $i$ of the two
    philosophers need not correspond to the same continuous
    time.  This table illustrates why it is necessary to
    account for continuous (real) time rather than logical
    time for the Generalised Dining Philosophers problem.
\label{tbl:2phil}}
\begin{center}
\begin{tabular}{rllll}
\hline
Time (sec) & Event & $A$'s state & $B$'s state & \\
\hline
0 &                         & $A[0] = \t$ & $B[0]= \t$ & \\
\hline
1.0 & $A1:\t\goesto{}{} \h$ & $A[1]  = \h$ &  & \\
\hline
2.0 & $A2:\h \goesto{}{} \e$ & $A[2] = \e$ &  & \\
\hline
5.0 & $B1:\t\goesto{}{} \h$ &  & $B[1] = \h$ & \\
\hline
7.0 & $A3:\e\goesto{}{}\t$ & $A[3] = \t$ &  & \\
\hline
8.0 & $B2:\h\goesto{}{}\e$ &  & $B[2] = \e$ & \\
\hline
11.0 & $A4:\t\goesto{}{}\h$ & $A[4] = \h$ &  & \\
\hline
12.0 & $A5:\t\goesto{}{}\e$ & $A[5] = \e$ &  & \\
\hline
\end{tabular}
\end{center}
\end{table}

\begin{figure}
\caption{Trajectories of two processes $A$ and $B$ over
  ordinal and physical time.\label{fig:2phil-unsynchronized}}
\centering{
 \includegraphics[width=5in]{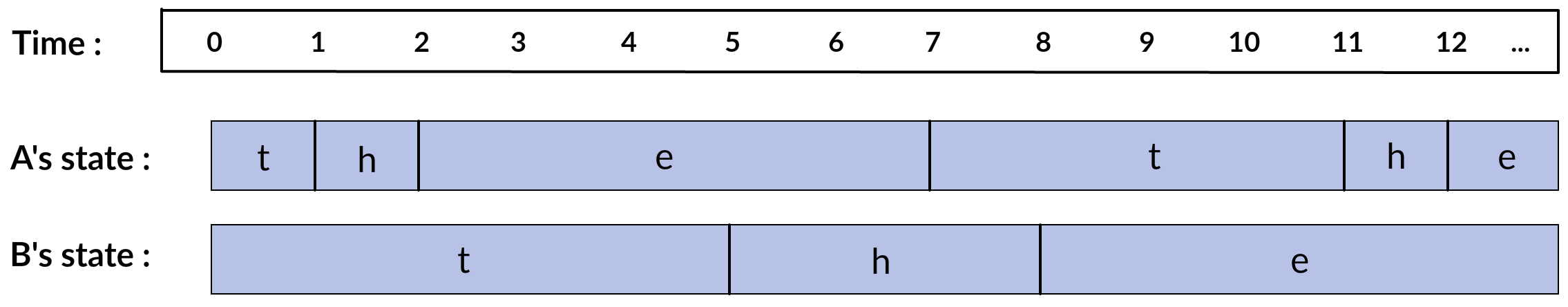}
}
\end{figure}

We employ the idea from time triggered architecture
\cite{Kopetz:1994:TPF:177213.177219} which is based on a
global clock with fixed time period against which all
subsystems transitions are synchronised, much like in a
hardware circuit.  The global clock ticks with a fixed time
period $\tau > 0$.  All systems share the global clock and
transitions are synchronised to occur in step with the
clock.  Values exist over continuous time, but are polled at
regular fixed intervals.  If $t\in \N$, $x[t]$ means the
value of $x$ at physical time $t\tau$.  Furthermore, the
transitions of all subsystems are synchronised: the $t$th
transition of each subsystem is transacted at exactly the
same physical time for each subsystem.  Transitions are
assumed to not be instantaneous.  If input is available at
clock cycle $t$, then the output of the new state as a
result of the transition is available only one clock cycle
later.

With clocked time, the composite system's dynamics may be
described as evolving over the same clock time as that of
the subsystems' dynamics.  Furthermore, since transitions
are not instantaneous but incur a delay, feedback becomes
easier to model.  Such a model has already been successfully
adopted by the family of synchronous reactive languages like
Esterel, Lustre and Signal\cite[Chap 2]{gamatie2009designing}.
Furthermore, extracting asynchronous behaviour becomes a
matter of making assumptions on the relative time periods
between input events and the computation of outputs.

\subsection{Clocked Systems}
The notion of a system is general enough to be able to model
a clock and also a system whose transitions are synchronised
with the ticks of the clock.  

\begin{notation}
\label{not:B}
Let \[B=\set{0,1}\]
\noindent denote a set of binary values and let $b$ range
over $B$.
\end{notation}

\subsubsection{Clock as a system}
\label{subsub:clock}
A clock ticking every $\tau$ units of time may be modelled
as a system $K$ whose state is time and whose input is an
arbitrary non-negative interval of time.  A state $x$
relates to $x'$ via time interval $u$ if $x'=x+u$.  The tick
is modelled as an impulse occurring at multiples of $\tau$.
\[K(\tau) = \pair{X,X^{0}, U,f:X,U\rightarrow X, Y, h}\]
where
\begin{itemize}
   \item $X=\Rnn$
   \item $X^{0}=\set{0}$
   \item $U=\Rp$
   \item $f(x,u) = x+u$
   \item $Y=B$
   \item $h(x) = 1$ if $x = n\tau$ for some $n\in \N$, $0$ otherwise.
\end{itemize}

\subsubsection{Extending a system interface to accommodate a clock}
\label{subsub:extend-system-with-clock}
To synchronise a system $S$ with a clock, it is first
necessary to extend the interface of the system to
accommodate an additional input of type $B$.  The extension
$S^{K}$ of $S$ is $S^{K}=\pair{X, X^{0}, U,
  \goesto{S^{K}}{}, Y, h}$ where
\begin{itemize}
 \item $X = X_S$
 \item $X^{0} = X^{0}_S$
 \item $U = U_S\times B$
 \item $x\goesto{S^{K}}{u_S,b}x'$ iff $x\goesto{S}{u_S}x'$ and
   $b=1$.
 \item $Y = Y_S$
 \item $h = h_S$ 
\end{itemize}
$S^{K}$ makes a transition only if it is admissible by the
underlying system $S$ \emph{and} its second input is 1.

\subsubsection{Synchronising a system with a clock}
\label{subsub:sync-system-with-clock}
The interconnect $\I$ wires the clock's output as the second
input of $S^{K}$:
\[\I = \set{(x_K, x_{S}, u_K, u_{S}, b)\ |\ h_K(x_K) = b}\]
In the composite system $T=K\times_{\I}S^{K}$, the
transitions of a system $S^{K}$ are now synchronised to
occur at each clock tick.  That is, the component
$x_{S^{K}}$ of $x_T$ is constant during the semi-open
interval $[i\tau, (i+1)\tau)$ and changes only at multiples
  of $\tau$.  Thus we may now treat $x_{S^{K}}$ (and
  $y_{S^{K}}$) as functions over $\N$, the set of naturals.
  Furthermore, inputs occurring at other than instances of
  the clock ticks effect no change of state.

From here on, we will not explicitly model the clock or its
composition with systems.  Instead, we assume that all
systems we design are implicitly clocked and there is one
global clock that drives all the subsystems of a system.

The dynamics of a system $S=\pair{X,X^{0},U,\goesto{S}{}, Y,
  h}$ suitably extended and interconnected with a clock may
now be described as a discrete dynamical system:
\begin{align}
  x[i] \goesto{S}{u[i]} x[i+1]\\
  y[i] = h(x[i])
\end{align}
\noindent where $i\in \N$ denotes the $i$th clock cycle and
$x$, $u$ and $y$ are functions from $\N$ to $X$, $U$, and
$Y$ respectively.

\begin{notation}
\label{not:prime}
For the sake of convenience, but at the risk of introducing
some ambiguity, we use the notation $x$ to mean $x[i]$ and
$x'$ to mean $x[i+1]$, the value at the next clock cycle.
Likewise for other variables.
\end{notation}

\section{The One Dining Philosopher problem}
\label{sec:1dp}
We start with N=1, the simplest case of the problem.  The 1
Diner problem is simple, but not trivial.  Indeed, as we
shall see, it reveals important insights about both the
problem structure and its solution for the general (N>1)
case.

We now systematise the formal construction of philosopher
system connected to a controller via feedback.  We start
with a philosopher model $Q$ that is completely
unconstrained in its behaviour, then build a deterministic
model $P$, identical in behaviour with $Q$, but in which
$Q$'s non-determinism is encoded as choice input.  $P$'s
interface is not quite suitable for participating in feedback
control.  That requires three more steps: First, extending
$P$ to the model $M$ which accommodates an additional
control input.  Second, defining a controller $C$ that
generates control input (\Cref{subsec:C}).  Third,
wiring the controller with the system $M$ to build a
feedback system $R$ (\Cref{subsec:R}).  Timing analysis
reveals that because of delays introduced in the feedback,
the control input may not arrive in time for it to be useful
(\cref{subsec:delays}).  A new input type where the signal
is present or absent and a plant $S$ working with this input
(\cref{subsec:S}) need to be composed with the controller in
such a way that the rate at which choice input arrives is
synchronised with the rate at which the controller computes
its output to yield the system $T$ (\cref{subsec:T}).

\subsection{Philosopher as an autonomous non-deterministic
  system}
\label{subsec:phil-Q}

An unconstrained philosopher (free to switch or stay) may be
modelled as an autonomous, non-deterministic, transparent
system 
\[Q = \pair{a:X=\Act, a^{0}:X^{0}=\set{\sf{t}},
  \goesto{Q}{}}\]

\noindent  where $\goesto{Q}{*}$ is defined via the
edges in \cref{fig:dp-states}:
 \begin{align*}
    \t\goesto{Q}{*}\t, \quad &\quad   \t\goesto{Q}{*}\h\\
    \h\goesto{Q}{*}\h, \quad &\quad   \h\goesto{Q}{*}\e\\
    \e\goesto{Q}{*}\e, \quad &\quad   \e\goesto{Q}{*}\t
 \end{align*}

\begin{notation}
\label{not:a}
We use the identifier $a$ to range over $\Act$.
\end{notation}

\subsection{Choice deterministic philosopher} 
\label{subsec:phil-P}
The non-determinism of $Q$ may be externalised by capturing
the choice at a state as binary input $b$ of type $B$ to the
system.

The resultant system is deterministic with respect to the
choice input.  On choice $b=0$ the system stays in the same
state; on $b=1$ it switches to the new state.  This is shown
below in the construction of a non-autonomous,
deterministic system
\[P =\pair{a:X=\Act, a^{0}:X^{0}=\set{\t}, b:U=B, f_P}\]
\noindent where $f_P:\Act\times B\rightarrow\ \Act$ is defined as
\begin{align}
f_P(a,0)  = \id{stay}(a)\label{eqn:f_P-eq1}\\
f_p(a,1)  = \id{switch}(a)\label{eqn:f_P-eq2}
\end{align}

\noindent and $\id{stay}:\Act\rightarrow \Act$ and
$\id{switch}:\Act\rightarrow \Act$ are given by
\begin{align*}
\id{stay}(a)    &= a\\
\id{switch}(\t) &= \h\\
\id{switch}(\h) &= \e\\
\id{switch}(\e) &= \t
\end{align*} 

The two systems $P$ and $Q$ are equivalent in behaviour.
\begin{prop}
$\B^{\omega}(Q) =  \B^{\omega}(P)$.
\end{prop}
\begin{proof}
For both systems, the output space and state space
are identical and the output functions are identity
functions.  Thus state and output traces are identical.

$\B^{\omega}(Q)\subseteq \B^{\omega}(P)$: For each state
trace $\overline{x}$ in $\B^{\omega}(Q)$, we construct an
input-state trace in $P$ and show that the corresponding
state trace in $P$ is $\overline{x}$.

For each $i\in \N$, let $a_i$ be the $i$th state in the
trace $\overline{a}$.  Then there is an input-state
transition $a\goesto{Q}{*}a'$.  If $a=a'$, then
we construct the transition $a\goesto{P}{0}a'$ of
$P$.  If $a\neq a'$, then we construct the transition
$a\goesto{P}{1}a'$. 

$\B^{\omega}(P)\subseteq \B^{\omega}(Q)$: For the
input-state transition $a\goesto{P}{u}a'$, we construct a
transition $a\goesto{Q}{*}a'$ in $Q$.

\end{proof}

\subsection{Interfacing control}
\label{subsec:M}

The Philosopher $P$ needs to be extended to admit control
input.  Control is accomplished through \emph{control input}
or \emph{command}:
\[\id{Cmd} = \set{\pass, \overzero, \overone}\]

\begin{notation}
\label{not:c}
We use the identifier $c:\id{Cmd}$ to denote a command. 
\end{notation}

The dynamics of a philosopher subject to a choice input
combined with a control input may be described as follows.
With command $c$ equal to $\pass$, the philosopher follows
the choice input $b$.  With the command equal to
$\override\;b$, the input is ignored, and the command
prevails in determining the next state of the philosopher
according to the value of $b$: stay if $b=0$, switch if
$b=1$.  

The philosopher system extended with a control input plays
the role of a model and is given by the transparent
deterministic system

\[M=\pair{a:X=\Act, a^{0}:X^{0}=\set{\t}, (b,c): (U_P\times U_F)=B \times\Cmd, f_M}\]

where $U_P$ denotes the preference (choice) and $U_F$
defines the forced (control) input and $f_M:\Act\times
(B\times\Cmd)\rightarrow \Act$ is given by
\begin{align*}
f_M(a, b, \pass) &= f_P(a,b)\\
f_M(a, \_, \override\;b) &= f_P(a,b)
\end{align*}
(In the second case, $\_$ indicates an unnamed formal
parameter whose name is not relevant because it is never
used subsequently.)


\subsection{Controller}
\label{subsec:C}

A controller is a transparent deterministic system $C$ whose
input is an activity and whose output is a control signal of
type $\id{Cmd}$.  The controller's role is to examine its
input and compute an output command based on the following
\emph{control law}: if its input is $\h$, then the output is
$\overone$, otherwise it is $\pass$\footnote{Other control
  laws are possible too.  As will be shown, the control law
  specified here is adequate to ensure the starvation freedom property
  for the N=1 Philosopher problem.}.  The controller's state
space is $\Cmd$ with initial state $\pass$ and its input $a$
is an activity.
\[C = \pair{c:X=\Cmd, c^{0}:X^{0}=\set{\pass}, a:U=\Act, f_C}\]
and $f_C:\Cmd\times \Act\rightarrow \Cmd$ is defined as
\begin{align*}
f_C(c,a) &=g_C(a)\\
g_C(\h) &= \overone\\
g_C(\e) &= \pass\\
g_C(\t) &= \pass 
\end{align*}

\subsection{Feedback composition}
\label{subsec:R}
Consider the interconnect $\IR\subseteq X_C\times X_M\times
U_C\times U_M$, between the controller $C$ and the model $M$
\[\IR = \set{(C.y, M.y, C.u, (M.b, M.c)) | M.y = C.u, C.y = M.c}\] 
$\IR$ specifies feedback composition since it connects the
philosopher's output $M.a$ to the input of the controller
$C.a$ and the controller's output $C.c$ to the control input
$M.c$ of the plant.  The composition $R=C\times_{\IR} M$ is
a deterministic system whose definition follows from the
definition of system composition.  We write $a$, $b$ and $c$
to denote the variables $R.M.a$, $R.M.b$ and $R.C.c$.








\Cref{fig:d1-feedback-system} shows a schematic of the
system $R$.

\usetikzlibrary{shapes,arrows}

\tikzstyle{block} = [draw, fill=blue!20, rectangle, 
    minimum height=3em, minimum width=6em]
\tikzstyle{sum} = [draw, fill=blue!20, circle, node distance=1cm]
\tikzstyle{input} = [coordinate]
\tikzstyle{output} = [coordinate]
\tikzstyle{point} = [coordinate]
\tikzstyle{pinstyle} = [pin edge={to-,thick,black}]

\begin{figure}
\caption{Feedback Control System $R$ for the 1 Diner problem\label{fig:d1-feedback-system}}
\centering{
\begin{tikzpicture}[auto, node distance=2cm,>=latex]
    \node [input, name=input] {};
    \node [block, right of=input, node distance=3cm] (controller) {$x_C$};
    \node [block, right of=controller, pin={[pinstyle]above:$u_P$},
            node distance=6cm] (model) {$x_M$};
    \draw [->] (controller) -- node[name=u] {$y_C$} (model);
    \node [output, right of=model, node distance=4cm] (output) {};
    \node [point, below of=u] (feedback) {feedback};

    \draw [draw,->] (input) -- node {} (controller);
    \draw [->] (model) -- node [name=y] {$y_M$}(output);
    \draw [->] (y) |- node[name=f] {} (feedback);
    \draw [-] (feedback) -| node[pos=0.01] {feedback}  node [near end] {} (input);
\end{tikzpicture}
} 
\end{figure}
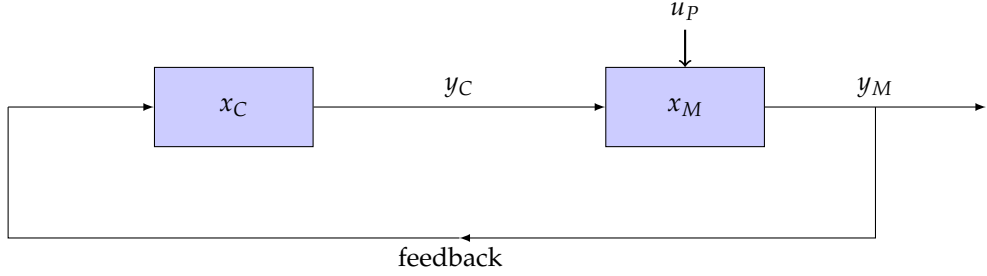

\subsection{Delays, Race conditions and Input rate}
\label{subsec:delays}
A simple example prefix run reveals a problem in the design
of the composite system $R$.
\Cref{tab:comp-for-prefix-1000} compares the desired and
actual behaviour of $R$ for the input choice sequence
$\pair{1,0,0,0}$.  One expects that the philosopher in state $\t$
at $t=1$, is commanded at $t=2$ to switch to $\e$ by the
controller.  However, the controller's output $\pass$ at
$t=2$ is computed based on the \emph{previous} philosopher
state at $t=1$, which was $\t$.  It takes one time step to
compute the control input, so the control input computed
is out of sync with the choice input.


\begin{table}
  \caption{Computations of state variables for the prefix
    $\pair{1,0,0,0}$ of choice input
    $b$\label{tab:comp-for-prefix-1000}.  The column labeled
    \id{desired} shows the expected value of the activity of
    the Philosopher constrained under the influence of a
    controller.  $R.a$ denotes the computed output of the
    subsystem when $M$ coupled with the output $R.c$ of the
    controller subsystem $C$.  Note that the computed
    behaviour $R.a$ does not match the desired
    behaviour. (The first mismatch is at clock cycle 2.)}

\begin{center}
\begin{tabular}{rrlll}
\hline
$t$ & $b$   & \id{desired} & $a'=f_M(a,b, c)$ & $c'=g_C(a)$\\
\hline\\
0 & $b^{0}=1$ & $\t$ & $\t= a^{0}$ & $c^{0}=\pass$\\
\hline
1 & 0 & $\h$ & $\h=f_M(\t,1,\pass)$ & $g_C(\t)=\pass$\\
\hline
2 & 0 & $\e$ & $\h=f_M(\h,0,\pass)$ & $g_C(\h)=\overone$\\
\hline
3 & 0 & $\e$ & $\e=f_M(\h,0,\overone)$ & $g_C(\h)=\overone$\\
\hline
4 & 0 & $\e$ & $\t=f_M(\e,0,\overone)$ & $g_C(\e)=\pass$\\
\hline
\end{tabular}
\end{center}
\end{table}

\begin{figure}
\caption{Graphical representation of trajectories in
  \cref{tab:comp-for-prefix-1000}.
  \label{fig:2phil}}
\centering{
 \includegraphics[width=5in]{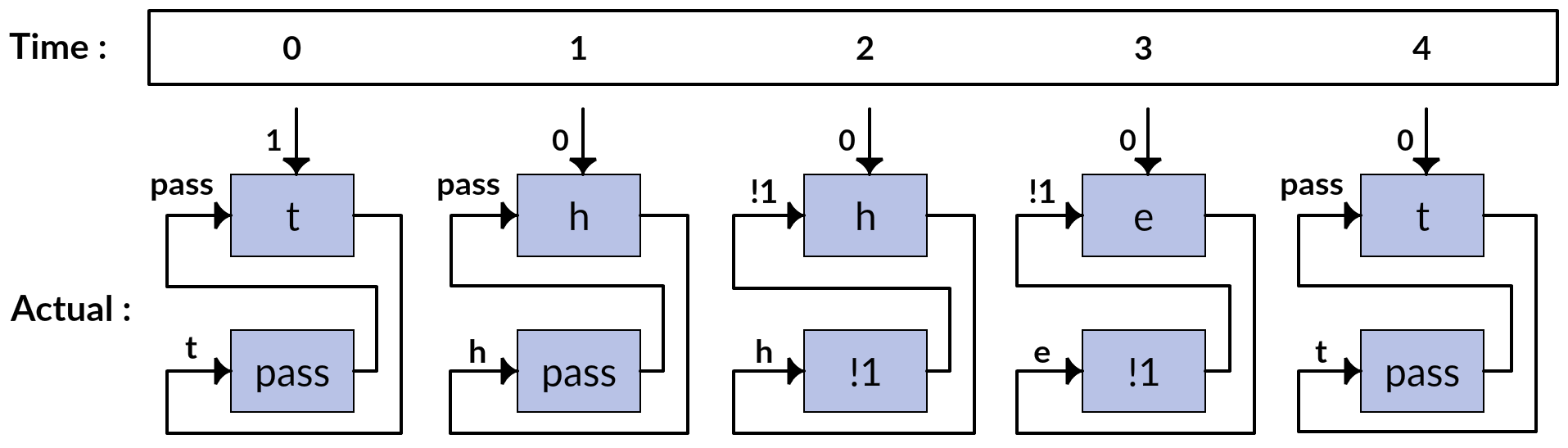}
}
\end{figure}

\subsection{Philosopher system with slower choice input}
\label{subsec:S}
In designing the controller and the new dynamics of the
philosopher, one needs to take into account the fact that
the controller needs one time step to compute its control
input.  During this step, no new input should arrive.  In
other words, the choice input should arrive slow enough
so that it synchronises with the arrival of the control input.

Keeping this in mind, we redesign the choice type to include
a $\bot$\ (read ``bottom'') input that denotes the absence
of choice.  This lifted input choice domain
\[\Bbot = \set{\bot}\union B\]
is now used to define the absence of input ($\bot$) or the
presence of a choice input (either 0 or 1).  We let the
variable $\bbot$ range over elements of $\Bbot$.  

A new deterministic and transparent philosopher system $S$
(for \emph{slower}) may then be defined as follows:

\[S=\pair{a:X=\Act, a^{0}:X^0=\set{\t}, (\bbot,c):U=\Bbot\times\Cmd, f_S}\]

\noindent where $f_S:\Act\times\Bbot\times\Cmd\rightarrow \Act$ is defined as 
\begin{align}
f_S(a, \bot, c) &= a\\
f_S(a, b, c)    &=f_M(a,b,c), \ \mbox{otherwise}
\end{align}

Expanding the definition of $f_M$, we have
\begin{align}
f_S(a, \bot, c) &= a\label{eqn:f_S-eq1}\\
f_S(a, b, \pass) &= f_P(a,b)\label{eqn:f_S-eq2}\\
f_S(a, \_, \override\;b)  &= f_P(a, b)\label{eqn:f_S-eq3}
\end{align}

If the choice input is $\bot$, the model $S$'s next state
stays the same as the previous state, irrespective of the
control input.  Otherwise, the $S$'s behaviour is just like
that of $M$: its next state is governed by the function
$f_M$, which expands to the two clauses $f_P$ shown above.

\subsection{System $T$:  feedback control system solving the
1 Diner problem}
\label{subsec:T}

The new composite system $T=C\times_{\I_{T}} S$ is defined
with respect to the interconnect 


\[\I_{T} =\set{(C.c,S.a, C.a, S.\bbot, S.c) | S.a=C.a, C.c=S.c}\]
\noindent which is similar to the interconnect $\IR$.  We
write $a$, $\bbot$ and $c$ to denote $S.a$, $S.\bbot$ and
$S.c$. 

In composing the system $S$ with the controller $C$, we
assume that the choice input to the philosopher alternates
between absent ($\bot$) and present (0 or 1).  In other
words, we assume that the choices are expressed slowly (with
one cycle of inactivity in between) so that the controller
has enough time to compute the control input.  (Another way
of achieving this is to drive the philosopher system with a
clock of time period of two units.)

{\bf Example} Consider the prefix of $T$'s behaviour on an
input choice stream with prefix
\[\pair{\bot,\ 1,\ \bot,\ 0,\ \bot,\ 0,\ \bot,\ 0,\ \bot,\ 1,\ \bot}\]

Note that each choice input is interspersed with one $\bot$.
The trace of $T$ shown in \cref{tab:1d-sync} demonstrates
that the discrepancy in \cref{tab:comp-for-prefix-1000} is
avoided.

\begin{table}
\caption{Computations of state variables in the system $T$
 for the prefix $\pair{\bot,\ 1,\ \bot,\ 0,\ \bot,\ 0,\ \bot,\ 0,\ \bot,\ 1,\ \bot}$
 of choice $u_{P'}$. \label{tab:1d-sync}}

\begin{center}
\begin{tabular}{rllll}
\hline
$t$ & $\bbot$ & desired & $a$ & $c$\\
\hline
 &  &  &  & \\
0 & \(\bot\) & \t & \t & \pass\\
\hline
1 & 1 & \t & \t & \pass\\
\hline
2 & \(\bot\) & \h & \h & \pass\\
\hline
3 & 0 & \h & \h & \overone\\
\hline
4 & \(\bot\) & \e & \e & \overone\\
\hline
5 & 0 & \e & \e & \pass\\
\hline
6 & \(\bot\) & \e & \e & \pass\\
\hline
7 & 0 & \e & \e & \pass\\
\hline
8 & \(\bot\) & \e & \e & \pass\\
\hline
9 & 1 & \e & \e & \pass\\
\hline
10 & \(\bot\) & \t & \t & \pass\\
\hline
\end{tabular}
\end{center}
\end{table}

\subsection{Dynamics of the 1 Diner system}
\label{subsec:analysis-1D}
We examine dynamics of the composite system $T$ with
philosopher subsystem $S$ interconnected with the controller
$C$.  Let $t\in \N$ denote the number of clock cycles of the
global clock whose time period is assumed one unit.  We
assume that each subsystem takes one clock cycle to compute
its next state given its input.  We also assume that
$\bbot[t]=\bot$ if $t$ is even, and equal to choice $b$,
where $b\in B$, if $t$ is odd.  

The following system of equations define the dynamics of the
1 Diner system:

\textbf{Initialisation:}

\begin{align*}
a[0] &= \t\\
\bbot[0] &= \bot\\
c[0] &= pass
\end{align*}

\textbf{Next state functions:}
\begin{align}
a[t+1] &= f_S(a[t], \bbot[t], c[t])\label{eqn:a-tp1-fs}\\
c[t+1] &= g_C(a[t])\label{eqn:c-tp1-gc}
\end{align}

Using the prime (') notation, these may be rewritten as
\begin{align}
a'     &= f_S(a, \bbot, c)\label{eqn:pa-tp1}\\
c'     &= g_C(a)\label{eqn:pc-tp1}
\end{align}

Given $\bbot[0]=\bot$, it is easy to verify that 
\begin{align*}
c[1]   &= c[0] = \pass\\
a[1]   &= a[0] = \t\\
\end{align*}

Tracing the dynamics from time $2t$ to $2t+3$, we have:

\begin{align*}
a[2t+1] &= f_S(a[2t], \bbot[2t], c[2t])\\
        &= f_S(a[2t], \bot, c[2t])& \mbox{From the defn. of $\bbot[2t]$} \\
        &= a[2t] \numberthis\label{eqn:a-2tp1-2t}\\\\
c[2t+1] &= g_C(a[2t])\\
        &= g_C(a[2t+1]) &\mbox{From \cref{eqn:a-2tp1-2t}}\numberthis\label{eqn:c-2tp1-gc}\\\\\\
a[2t+2] &= f_S(a[2t+1], \bbot[2t+1], c[2t+1])\\
        &= f_S(a[2t+1], b[2t+1], c[2t+1])&
\mbox{From the defn. of $\bbot[2t+1]$}\numberthis\label{eqn:a-2tp2-fs}\\\\
c[2t+2] &= g_C(a[2t+1])\\
        &= g_C(a[2t])&\mbox{From \cref{eqn:a-2tp1-2t}}\\
        &= c[2t+1]\numberthis\label{eqn:c-2tp2-2tp1}\\\\\\
a[2t+3] &= f_S(a[2t+2], \bbot[2t+2], c[2t+2])\\
        &= f_S(a[2t+2], \bot, c[2t+2])\\
        &= a[2t+2] & \mbox{From the defn. of $f_S$ (\cref{eqn:f_S-eq1})}\numberthis\label{eqn:a-2tp3-2tp2}\\
        &= f_S(a[2t+1], b[2t+1], c[2t+1]) &
\mbox{From
  \cref{eqn:a-2tp2-fs}}\numberthis\label{eqn:a-2tp3-fs}\\\\
c[2t+3] &= g_C(a[2t+2])\\
        &= g_C(a[2t+3]) & \mbox{From \cref{eqn:a-2tp3-2tp2}}\\
\end{align*}

From this we conclude the following, for $t\in \N$.

\begin{align*}
a[2t+3] &= f_S(a[2t+1], b[2t+1], c[2t+1])&\mbox{Ref. \cref{eqn:a-2tp3-fs}}\\
\bbot[2t+1] &= b[2t+1] &\mbox{Assumption}\\
c[2t+1] &= g_C(a[2t+1])&\mbox{Ref. \cref{eqn:c-2tp1-gc}}
\end{align*}

\subsection{Simplified dynamics by polling}
\label{subsec:simplified-dynamics-1diner}
The dynamics may be reduced to a simpler system of equations
if we consider polling the system once every two clock
cycles.  We define a \emph{step} to be two clock cycles,
with the $i$th step corresponding to the $2i+1$ clock cycle.
The relation between the new set of variables $[a_2, b_2,
  c_2]$ and the previous variables is shown
below\footnote{We have assumed $\bbot[0]=\bot$.  If we
  assumed that $\bbot[0]=b[0]$, then the equations would be
  $a_2[i]=a[2i]$, etc.}:
\begin{align*}
a_2[0]   &= a[1] =\t\\
c_2[0]   &= c[1] =\pass\\
b_2[0]   &= \bbot[1] =b[1]
\end{align*}
and
\begin{align*}
a_2[i]   &= a[2i+1]\\
b_2[i]   &= \bbot[2i+1]\\
c_2[i]   &= c[2i+1]
\end{align*}

To continue using the old variables, we abuse notation and
write $a$ etc., to refer to $a_2$, etc.  Thus the polled
dynamics, indexed over steps $i$ reduces to:
\begin{align*}
a[0]   &= \t\\
c[0]   &= \pass\\
c[i]   &= g_C(a[i])\\
a[i+1] &= f_S(a[i], b[i], c[i])
\end{align*}

We simplify notation further by making the indexing with $i$
implicit and writing $a$ to mean $a[i]$ and $a'$ to denote
$a[i+1]$.  Thus
\begin{align}
a^0    &= t\label{eqn:az}\\
c^0    &= \pass\label{eqn:cz}\\
c      &= g_C(a)\label{eqn:c}\\
a'     &= f_S(a, b, c)\label{eqn:ap}
\end{align}

\Crefrange{eqn:az}{eqn:ap} completely capture the 'polled
dynamics' of the composite system consisting of the
controller with the philosopher.  Note that $\bot$ is no
longer relevant to the polled dynamics.

\subsection{Correctness of the solution for the 1 Diner problem}
\label{subsec:1d-correctness}

\begin{prop}
  Consider the composite system $T = C\times_{\I}S$ working
  under the assumption that choice inputs arrive only at
  odd cycles.  Then, the system correctly implements the
  starvation freedom constraint of the 1 Diner problem which states
  that the philosopher doesn't remain hungry forever.  It
  defers to the philosopher's own choice (stay at the same
  state or switch to the next) when the philosopher is not
  hungry.
\end{prop}
\begin{proof}
The result follows from the following propositions, which
are simple consequences of the polled dynamics: 
\be
\item \label{itm:hungry} if $a=\h$, then $a'=\e$.
\item \label{itm:other} if $a\neq \h$, then $a'=f_P(a,b)$. 
\ee

\end{proof}

\section{N Dining Philosophers with Centralised control}
\label{sec:ndp}
We now look at the Generalised Dining Philosophers
problem.  We are given a graph $G=\pair{V,E}$, with $|V|=N$
and with each of the N vertices representing a philosopher
and $E$ representing an undirected, adjacency relation
between vertices.  The vertices are identified by integers
from 1 to N.  

Each of the N philosophers are identical and modeled as the
instances of the system $S$ described in the 1 Diner case.
These N vertices are all connected to a single controller
(called the hub) which reads the activity status of each of
the philosophers and then computes a control input for that
philosopher.  The control input, along with the choice input
to each philosopher computes the next state of that
philosopher.

\begin{figure}
\caption{Wiring diagram describing the architecture of
  centralised
  controller.\label{fig:centralised-n-diners-arch}}
\centering{
 \includegraphics[width=5in]{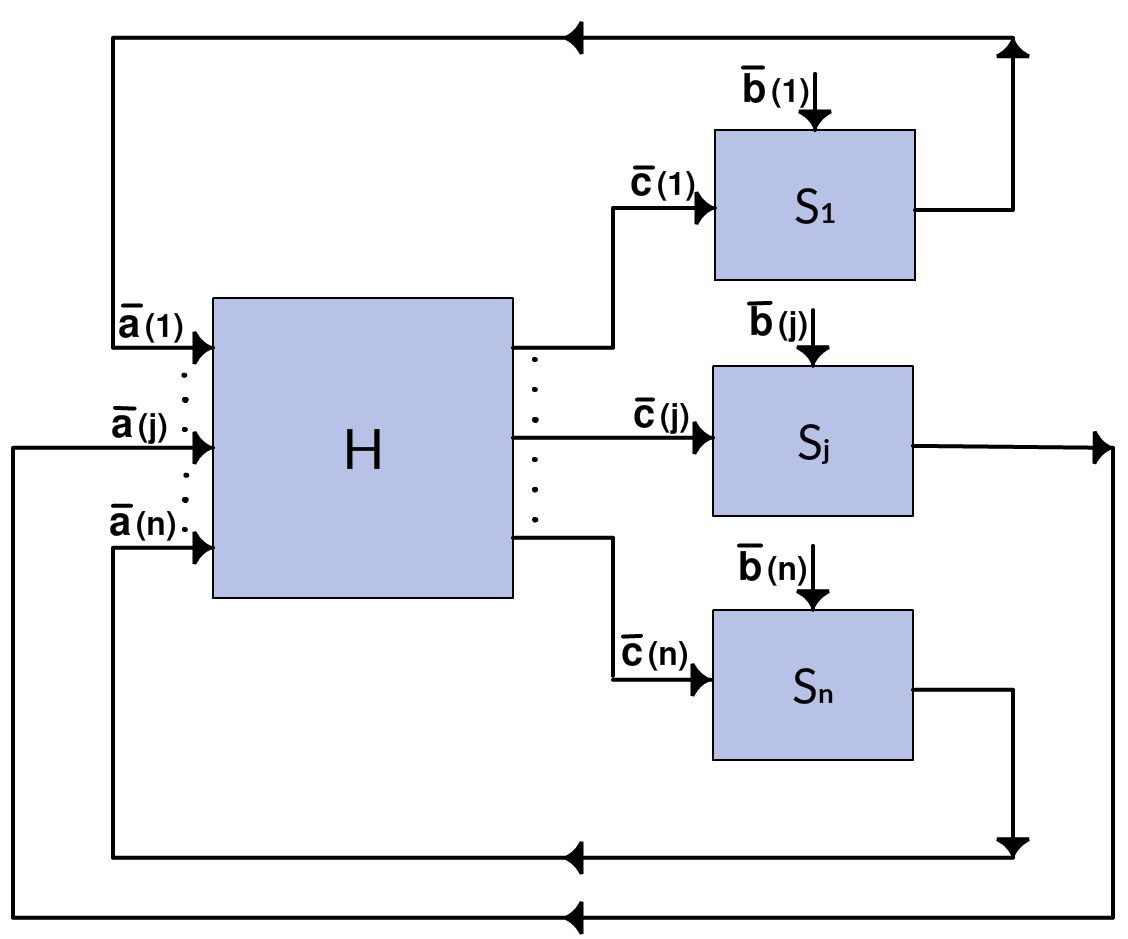}
}
\end{figure}

\begin{notation}
\label{not:maps}

Identifiers $j,k,l\in V$ denote vertices.

An \emph{activity map} $\oa: V\rightarrow A$ maps 
vertices to their status, whether hungry, eating or
thinking. 

A \emph{choice map} $\ob: V\rightarrow B$ maps to
each vertex a choice value.  

A \emph{maybe choice map} $\om: V\rightarrow \Bbot$ maps to
each vertex a maybe choice value (nil or a choice). 

A \emph{command map} $\oc: V\rightarrow \id{Cmd}$ maps to
each vertex a command.

If $v$ is a constant, then $\overline{v}$ denotes a function
that maps every vertex to the constant $v$.
\end{notation}

The data structures and notation used in the solution are
described below:

\be
 \item $G=(V,E)$, the graph of vertices $V$ and their
   adjacency relation $E$.  $G$ is part of the hub's
   internal state.  $G$ is constant throughout the problem. 

   We write $\set{j,k}\in E$, or $E(j,k)$ to denote that
   there is an undirected edge between $j$ and $k$ in $G$.
   We write $E(j)$ to denote the set of all neighbours of
   $j$.

 \item $\oa:V\rightarrow \set{\t,\h,\e}$, an activity map.
   This is input to the hub controller.

 \item $D: (j,k)\in E\rightarrow \set{j,k}$, is a directed
   relation derived from $E$.  $D$ is called a
   \emph{dominance map} or \emph{priority map}.  For each
   edge $\set{j,k}$ of $E$ it returns the source of the
   edge.  The element $\set{j,k}\mapsto j$ of $D$ is
   indicated $j\mapsto k$ ($j$ dominates $k$) whereas
   $\set{j,k}\mapsto k$ is indicated $k\mapsto j$ ($j$ is
   dominated by $k$).  If $\set{j,k}\in E$, then exactly one
   of $j\mapsto k\in D$ or $k\mapsto j\in D$ is true.

  $D(j)$ is the set of vertices dominated by $j$ in $D$ and
   is called the set of \emph{subordinates} of $j$.  $D^{-1}(j)$
   denotes the set of vertices that dominate $j$ in $D$ and
   is called the set of \emph{dominators} of $j$.

 \item $\id{top}(D)$, the set of maximal elements of $D$.
   $\id{top}(D)(j)$ means that $j\in \id{top}(D)$.  This is
   a derived internal state of the hub controller.

 \item $\oc:V\rightarrow \id{Cmd}$, the command map.  This
   is part of the internal state of the hub controller and
   also its output.  
\ee

\paragraph{Additional Notation} 
Let $s\in A$, and $\oa$ be an activity map.
Then $E_{s}(\oa)(j)$ denotes the set of neighbours of $j$
whose activity value is $s$.  Likewise $D_{s}(\oa)(j)$
denotes the set of vertices in the subordinate set of $j$
whose activity status is $s$.  

\subsection{Informal introduction to the control algorithm}

Initially, at cycle $t=0$, all vertices in $G=(V,E)$ are
thinking, so $\oa[0]=\ot$.  Also, $D[0]$ is $D^{0}$,
$\id{top}(D)[0] = \{j\ |\ D^{0}(j) = E(j)\}$ and $\oc[0]=\opass$.

Upon reading the activity map, the controller performs the
following sequence of computations:

 \be

\item (Step 1): Updates $D$ so that (a) a vertex that is
  eating is dominated by all its neighbours, and (b) any
  hungry vertex also dominates its thinking neighbours. 

\item (Step 2): Computes $\id{top}$, the set of top
  vertices.

   \item (Step 3): Computes the new control input for each
     philosopher vertex: A thinking or eating vertex is
     allowed to pass.  A hungry vertex that is at the top
     and has no eating neighbours is commanded to switch to
     eating.  Otherwise, the vertex is commanded to stay
     hungry.  
\ee

\subsection{Formal structure of the Hub controller}

The centralised or hub controller is a deterministic system 
$H=\pair{X, X^{0}, U, f, Y, h}$, where
\be
  \item $X_H=(E\rightarrow \Bool) \times (V\rightarrow
    \id{Cmd})$ is the cross product of the set of all
    priority maps derived from $E$ with the set of command
    maps on the vertices of $G$.  Each element $x_H:X_H$ is
    a tuple $(D, \oc)$ consisting of a priority map $D$ and
    a command map $\oc$.

  \item $X^{0}_H=(D^{0},\oc^{0})$ where $D^{0}(\set{j,k}) =
    j\mapsto k$ if $j>k$ and $k\mapsto j$ otherwise for
    $\set{j,k}\in E$, and $\oc^{0}(j) = \opass$.  Note that
    $D^{0}$ is acyclic.

  \item $U_H$ is the set of activity maps.  $\oa:U_H$
    represents the activity map that is input to the hub $H$.

  \item $f_H:X,U\rightarrow X$ takes a priority map $D$, a
    command map $\oc$, and an activity  map $\oa$ as input
    and returns a new priority map $D'$ and a new command
    map.

$f_H((D,\oc),\oa) = (D', g_H(D',\oa))$  where
\begin{align}
D'                    &= d_H(D,\oa)\\
d_H(D,\oa)              &\ida \set{d_H(d,\oa)\ |\ d \in D}\\
d_H(j\mapsto k,\oa)     &\ida (k\mapsto j), \quad \mbox{if $\oa(j)=\e$}\label{eqn:dh1}\\
                      &\ida (k\mapsto j), \quad \mbox{if $\oa(j)=\t$\ and\ $\oa(k)=\h$}\label{eqn:dh2}\\
                      &\ida (j\mapsto k), \quad \mbox{otherwise}\label{eqn:dh3}
\end{align}

Note that the symbol $d_H$ is overloaded to work on a
directed edge as well as a priority map.  $d_H$ implements
the updating of the priority map $D$ to $D'$ mentioned in
(Step 1) above.  The function $g_H$ computes the command map
(Step 3).  The command is \id{\pass} if $j$ is either eating
or thinking.  If $j$ is hungry, then the command is
\overone if $j$ is ready, i.e., it is hungry, at the top
(Step 2), and its neighbours are not eating.  Otherwise, the
command is \overzero.

\begin{align}
   g_H(D,\oa)(j)    &\ida \pass, \quad \mbox{if $\oa(j)\in \set{\t,\e}$}\label{eqn:g_H}\\
                    &\ida \overone, \quad \mbox{if $\id{ready}(D,\oa)(j)$} \\
                    &\ida \overzero, \quad \mbox{otherwise} \\ \nonumber \\
  \id{ready}(D,\oa)(j) \ida \true, \quad \mbox{if } & \oa(j)=\h\ \land\\
                             & j\in \id{top}(D) \ \land \nonumber \\
                             & \forall k\in
  E(j):\ \oa(k)\neq\e \nonumber \\ \nonumber \\
  \id{top}(D)        &\ida   \set{j\in V\| \forall k\in E(j): j\mapsto k}
 \end{align}

\item $Y_H = V\rightarrow \id{Cmd}$: The output is a command
  map.

\item $h_H:X_H\rightarrow Y_H $ simply projects the command
  map from its state: $h_H(D,\oc) \ida \oc$.

\ee 

  Note that an existing priority map $D$ when combined with
  the activity map results in a new priority map $D'$.  The
  new map $D'$ is then passed to $g_H$ in order to compute
  the command map.

The first important property concerns the priority map
update function. 
\begin{lemma}[$d_H$ is idempotent]
\label{lem:dh-idempotent}
$d_H(D,\oa)=d_H(d_H(D,\oa),\oa)$.
\end{lemma}
\begin{proof}
The proof is a simple consequence of the definition of
$d_H$. 
\end{proof}

\subsection{Composing the hub controller with the
  Philosophers}
\label{subsec:hub-composition}

Consider the interconnect $\I$ between the hub $H$ and the
$N$ philosopher instances $s_j$, $1\leq j\leq N$.

\[\I\subseteq X_H\times U_H\times \Pi_{j=1}^{N}s_j.X\times s_j.U\] 

\noindent that connects the output of each philosopher to
the input of the hub, and connects the output of the hub to
control input of the corresponding philosopher.

\begin{align*}
\I = & \{(x_H, u_H, s_1.x, s_1.u \ldots s_n.x,s_n.u)\ |\\
   & u_H(j) = h_S(s_j.x)\ \land\ h_H(x_H)(j) = s_j.u,\\
   & 1\leq j\leq N\}
\end{align*}

The composite N Diners system is the product of the N+1
systems.  

We assume that the composite system is synchronous and
driven by a global clock.  At time $i$, the activity map
$\oa[i]$ holds the $j$th philosopher's activity at
$\oa[i](j)$.  All the philosophers make their choice inputs
at the same instant and the choice inputs alternate with the
$\bot$\ inputs.  Without loss of generality, we assume that
the controller takes one clock cycle to compute the control
input.

The dynamics of the entire system may be described by the
following system of equations:

\textbf{Initialisation:}

\begin{align*}
D[0]   &= D^{0}\\
\oc[0] &= \opass\\
\oa[0] &= \ot\\
\obb[0] &= \obot\\
\end{align*}

\textbf{Next state functions:}
\begin{align}
D[t+1]   &= d_H(D[t],\oa[t])\label{eqn:Dtp1}\\
\oc[t+1] &= g_H(D[t+1], \oa[t])\label{eqn:octp1}\\
\oa[t+1] &= f_S(\oa[t], \obb[t], \oc[t])\label{eqn:oatp1}
\end{align}

Using the prime (') notation, these may be rewritten as
\begin{align}
D'       &= d_H(D,\oa)\label{eqn:pDtp1}\\
\oc'     &= g_H(D', \oa)\label{eqn:poctp1}\\
\oa'     &= f_S(\oa, \obb, \oc)\label{eqn:poatp1}
\end{align}

The input to the system, the lifted choice map $\obb$
alternates between $\obot$ at time $2t$ and a choice map
$\ob$ at time $2t+1$.

Given $\obb[0]=\obot$, it is easy to verify that 
\begin{align*}
D[1]     &= D[0]   =D^{0}\\
\oc[1]   &= \oc[0] =\opass\\
\oa[1]   &= \oa[0] =\ot
\end{align*}

Tracing the dynamics from time $2t$ to $2t+3$, we have:

\begin{align*}
\oa[2t+1] &= f_S(\oa[2t], \obb[2t], \oc[2t])\\
          &= f_S(\oa[2t], \obot[2t], \oc[2t])& \mbox{From
  the defn. of $\obb[2t]$} \\
          &=    \oa[2t] \numberthis\label{eqn:os-2t}
\end{align*}

\begin{align*}
D[2t+1]  &= d_H(D[2t], \oa[2t])\numberthis\label{eqn:D-2tp1}
\end{align*}

\begin{align*}
\oc[2t+1] &= g_H(D[2t+1], \oa[2t])\\
          &= g_H(D[2t+1], \oa[2t+1]) &\mbox{From \cref{eqn:os-2t}}\numberthis\label{eqn:oc-2tp1}
\end{align*}

\begin{align*}
\oa[2t+2] &= f_S(\oa[2t+1], \obb[2t+1], \oc[2t+1])\\
          &= f_S(\oa[2t+1], \ob[2t+1], \oc[2t+1])&
\mbox{From the defn. of $\obb[2t+1]$}\numberthis\label{eqn:os-2tp2}
\end{align*}

\begin{align*}
D[2t+2]   &= d_H(D[2t+1], \oa[2t+1])\\
          &= d_H(D[2t+1], \oa[2t]) &\mbox{From \cref{eqn:os-2t}}\\
          &= D[2t+1] &\mbox{From \labelcref{eqn:D-2tp1} and idempotence of $d_H$}\numberthis\label{eqn:D-2t}
\end{align*}

\begin{align*}
\oc[2t+2] &= g_H(D[2t+2], \oa[2t+1])\\
          &=g_H(D[2t+2], \oa[2t])&\mbox{From \cref{eqn:os-2t}}\\
          &=g_H(D[2t+1], \oa[2t])&\mbox{From \cref{eqn:D-2t}}\\
          &=\oc[2t+1]\numberthis\label{eqn:oc-2tp2}
\end{align*}

\begin{align*}
\oa[2t+3] &= f_S(\oa[2t+2], \obb[2t+2], \oc[2t+2])\\
          &= f_S(\oa[2t],   \obot, \oc[2t+1])\\
          &= \oa[2t+2] & \mbox{From the defn. of $f_S$ (\cref{eqn:f_S-eq1})}\numberthis\label{eqn:os-2tp3-2tp2}\\
          &= f_S(\oa[2t+1],   \ob[2t+1], \oc[2t+1]) &
\mbox{From
  \cref{eqn:os-2tp2}}\numberthis\label{eqn:os-2tp3}
\end{align*}

\begin{align*}
D[2t+3]  &= d_H(D[2t+2], \oa[2t+2])\\
         &= d_H(D[2t+1], \oa[2t+3]) & \mbox{From \cref{eqn:os-2tp3-2tp2}}\\\\
\oc[2t+3] &= d_H(D[2t+3], \oa[2t+2])\\
          &= d_H(D[2t+3], \oa[2t+3]) & \mbox{From \cref{eqn:os-2tp3-2tp2}}
\end{align*}

From this we conclude the following, for $t\in \N$.

\begin{align*}
\obb[2t+1] &= \ob[2t+1] &\mbox{Assumption}\\
\oc[2t+1] &= g_H(D[2t+1],\oa[2t+1])&\mbox{Ref. \cref{eqn:oc}}\\\\
\oa[2t+3] &= f_S(\oa[2t+1],   \ob[2t+1], \oc[2t+1])&\mbox{Ref. \cref{eqn:os-2tp3}}\\
D[2t+3]   &= d_H(D[2t+1], \oa[2t+3])&\mbox{Ref. \cref{eqn:os-2tp3-2tp2}}
\end{align*}

\subsection{Simplified dynamics by polling}
\label{subsec:simplified-dynamics}
The dynamics may be reduced to a simpler system of equations
if we consider polling the system once every two clock
cycles.  We consider a new clock of twice the time period.
The index variable $i$ refers to the newer clock.  The
relation between the new set of variables $[\oa_2, \ob_2,
  \oc_2, D_2]$ and  the previous variables is shown below:

\begin{align*}
\oa_2[i]   &= \oa[2i+1]\\
\ob_2[i]   &= \obb[2i+1]\\
\oc_2[i]   &= \oc[2i+1]\\
D_2[i]     &= D[2i+1]
\end{align*}

and 

\begin{align*}
\oa_2[0]   &= \oa[1] =\ot\\
\oc_2[0]   &= \oc[1] =\opass\\
D_2[0]     &= D[1]   =D^{0}\\
\ob_2[0]   &= \obb[1] =\ob[1]\\
\end{align*}

To continue using the old variables, we abuse notation and
write $\oa$ etc., to refer to $\oa_2$, etc.  Thus the
dynamics based on the new clock with ticks indicated by $i$
is shown below:


\begin{align*}
\oa[0]   &= \ot\\
D[0]     &= D^0\\
\oc[0]   &= \opass\\
\oc[i]   &= g_H(D[i],\oa[i])\\
\oa[i+1] &= f_S(\oa[i], \ob[i], \oc[i])\\
D[i+1]   &= d_H(D[i], \oa[i+1])\\
\end{align*}

We simplify notation further by making the indexing with $i$
implicit and writing $\oa$ to mean $\oa[i]$ and $\oa'$ to
denote $\oa[i+1]$.  Thus

\begin{align}
\oa^0    &= \ot\label{eqn:oaz}\\
\oc^0    &= \opass\label{eqn:ocz}\\
D^0      &= \set{j\mapsto k\ |\ E(j,k) \land j > k}\label{eqn:Dz}\\[2em]
\oc      &= g_H(D,\oa)\label{eqn:oc}\\
\oa'     &= f_S(\oa, \ob, \oc)\label{eqn:oap}\\
D'       &= d_H(D, \oa')\label{eqn:Dp}
\end{align}





  






\Crefrange{eqn:oaz}{eqn:Dp} completely capture the 'polled
dynamics' of the composite system consisting of the hub
controller with the N Diners.  This dynamics is obtained by
polling all odd instances of the clock, which is precisely
when and only when the choice input is present.  With the
polled dynamics, we are no longer concerned with $\bot$ as
as a choice input.

It is worth comparing the polled dynamics with the basic
clocked dynamics of \crefrange{eqn:pDtp1}{eqn:poatp1}.
Note, in particular, the invariant that relates $\oa$, $\oc$
and $D$ in \cref{eqn:oc} of the polled dynamics.  There is
no such invariant in the basic clocked dynamics.
\Cref{eqn:oap} of the polled dynamics may be seen as a
specialisation of the corresponding \cref{eqn:poatp1} of the
basic clocked dynamics.  However, while \cref{eqn:Dp} of the
polled dynamics relates $D'$ ($D$ in the next step) with $D$
and $\oa'$, its counterpart \cref{eqn:pDtp1} in the basic
clocked dynamics relates $D'$ ($D$ in the next cycle) with
$D$ and $\oa$.

\subsection{Asynchronous interpretation of the dynamics}
\label{subsec:async}

It is worth noting that the equations we obtained in the
polled dynamics of the system can be interpreted as
asynchronous evolution of the philosopher system.  A careful
examination of the equations yields temporal dependencies
between the computations of the variables involved in the
systems.  Consider the polled equations, consisting of
indexed variables $\oa$, $\oc$ and $D$:

\begin{align*}
\oa[0]   &= \ot\\
D[0]     &= D^0\\
\oc[i]   &= g_H(D[i],\oa[i])\\
\oa[i+1] &= f_S(\oa[i], \ob[i], \oc[i])\\
D[i+1]   &= d_H(D[i], \oa[i+1])\\
\end{align*}

The asynchronous nature of the system dynamics tells us that
the $i^{th}$ value of $\oc$ requires the $i^{th}$ values of
$\oa$ and $D$ to be computed before its computation happens,
and so on.  This implicitly talks about the temporal
dependency of the $i^{th}$ value of $\oc$ on the $i^{th}$
values of $\oa$ and $D$.  Similarly, the $(i+1)^{th}$ value
of $\oa$ depends on the $i^{th}$ values of $\oa$, $\oc$ and
$D$, and the $(i+1)^{th}$ value of $D$ depends on the
$i^{th}$ value of $D$ and the $(i+1)^{th}$ value of $\oa$.
Note that they only talk about the temporal dependencies
between variable calculations, and do not talk about the
clock cycles, nor when the values are computed in physical
time.  The following figure depicts the dependencies between
the variables.

\begin{figure}[H]
\caption{Dependencies between $\oa$, $\oc$ and $D$, along
  with input $\ob$, shown for three calculations.\label{fig:temporal-dependencies}}
\centering{
 \includegraphics[width=5in]{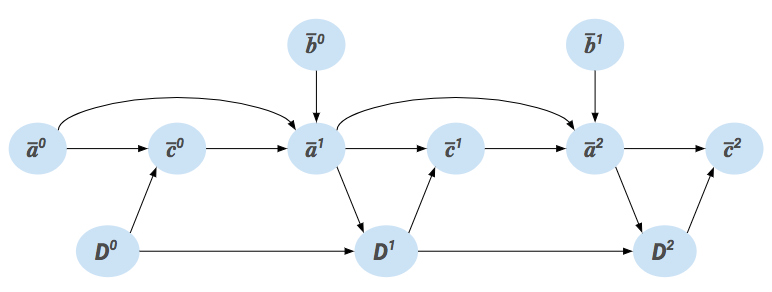}
}
\end{figure}

\subsection{Basic properties of asynchronous dynamics}
\label{subsec:polled-dynamics}

In the next several lemmas, we study the asynchronous (or
polled) dynamics in detail.  All of these are simple
consequences of the functions $g_H$ and $f_S$.  The first of
several lemmas in this effort assures us that the
asynchronous dynamics obeys the laws governing the dynamics
of basic philosopher activity:
\begin{lemma}[Asynchronous Dynamics]
\label{lem:t-next}
Let $k\in V$.  The dynamics satisfies the following:
\be
 \item If $\oa(k)=\t$, then $\oa'(k)\in \set{\t,\h}$.
 \item If $\oa(k)=\h$, then $\oa'(k)\in \set{\h,\e}$.
 \item If $\oa(k)=\e$, then $\oa'(k)\in \set{\e,\t}$.
\ee
\end{lemma}
\begin{proof}

This is a consequence of the dynamics $\oa'=f_S(\oa,
\ob,\oc)$ and simply substituting the definitions of $\oc$
and the value of $\oa(k)$.  We show the case when
$\oa(k)=\t$.  The others are similar.

\begin{align*}
\oa'(k) &= f_S(\oa, \ob,\oc)(k)\\
        &= f_S(\t, \ob(k),\oc(k))\\
        &= f_S(\t, \ob(k), g_H(D,\oa)(k))\\
        &= f_S(\t, \ob(k),\pass)   &\mbox{From the defn. of
  $g_H$ (\cref{eqn:g_H})}\\
        &= f_P(\t, \ob(k)) &\mbox{From the defn. of $f_S$ (\cref{eqn:f_S-eq2})}\\
        &= \t,  &\mbox{if $\ob(k)=0$, or}\\
        &= \h,  &\mbox{if $\ob(k)=1$}
\end{align*}
\end{proof}

The next lemma invests meaning to the phrase ``priority map''.
If $j$ has higher priority than a hungry vertex $k$,
irrespective of whether $j$ stays hungry or switches to
eating, in the next step, $k$ will continue to wait in the
hungry state.
\begin{lemma}[Priority map]
\label{lem:D-dom}
If $j\mapsto k\in D$ and $\oa(k)=\h$, then $\oa'(k)=\h$. 
\end{lemma}

\begin{proof}
\begin{align*}
\oa'(k) &= f_S(\h, \ob(k),\oc(k))\\
        &= f_S(\h, \ob(k), g_H(D,\oa)(k))\\
        &= f_S(\h, \ob(k),\overzero)    &\mbox{Since $\id{ready}(D,\oa)(k)$ is false}\\
        &= f_P(\h, 0) &\mbox{From the defn. of $f_S$
(\cref{eqn:f_S-eq3})}\\
        &= \h  &\mbox{From the defn. of $f_P$ (\cref{eqn:f_P-eq1})}\\
\end{align*}
\end{proof}

The next lemma states that a hungry vertex with an eating
neighbour stays hungry in the next step,
irrespective of the eating vertex finishing eating in the
next step or not.  This lemma drives the safety
invariant (described later) that ensures that no two
adjacent vertices eat at the same time.

\begin{lemma}[Continue to be hungry if neighbour eating]
\label{lem:h-if-n-eats}
If $E(j,k)$,  $\oa(j)=\h$ and $\oa(k)=\e$ then $\oa'(j)=\h$. 
\end{lemma}
\begin{proof}
\begin{align*}
\oa'(j) &= f_S(\oa(j), \ob(j),\oc(j))\\
        &= f_S(\h, \ob(j), g_H(D,\oa)(j))\\
        &= f_S(\h, \ob(j),\overzero)    &\mbox{Since $\id{ready}(D,\oa)(j)$ is false}\\
        &= f_P(\h, 0) &\mbox{From the defn. of $f_S$
(\cref{eqn:f_S-eq3})}\\
        &= \h  &\mbox{From the defn. of $f_P$ (\cref{eqn:f_P-eq1})}
\end{align*}
\end{proof}

\subsection{Safety and other invariants}
\label{subsec:safety-invariants}

\begin{theorem}
\label{thm:safety-invariants}

The dynamics of the composition of the N philosophers with
the hub controller satisfies the following invariants:

  \be

  \item\label{itm:e-sink} \textbf{Eaters are sinks:} If
    $\oa(k)=\e$ and $E(j,k)$, then $D(j,k)$.

  \item\label{itm:h-dom-t} \textbf{Hungry dominate
    thinkers:} If $\oa(j)=\h$ and $\oa(k)=\t$, and $E(j,k)$,
    then $D(j,k)$.

  \item\label{itm:safe} \textbf{Safety:} $\id{safe}(E,\oa)$:
    $\oa(j)=\e$ and $E(j,k)$ implies $\oa(k)\neq\e$.  
\ee

\end{theorem}

\begin{proof}

The proof is by induction on $i$.
\be

 \item \textbf{Eaters are sinks:} 

  The base case is vacuously true since $\oa^{0}=\ot$.

  For the inductive case, assume $\oa'(k)=\e$ and $E(j,k)$,
  we need to show that $D'(j,k)$.  This follows from the
  definition $D'=d_H(D,\oa')$ and from the definition of
  $d_H$ (clause \labelcref{eqn:dh1}).

 \item \textbf{Hungry dominate thinkers:}

  The proof of this claim is similar to that of the previous
  claim.

  The base case is vacuously true since $\oa^{0}=\ot$ (there
  are no hungry nodes). 

  For the inductive case, assume $\oa'(j)=\h$ and $E(j,k)$
  and $\oa'(k)=\h$, we need to show that $D'(j,k)$.  Now,
  $D'=d_H(D,\oa')$.  From the definition of $d_H$ (clause
  \labelcref{eqn:dh2}), it follows that $j\mapsto k\in D'$,
  i.e., $D'(j,k)$.

 \item \textbf{Safety:}

 The base case is trivially true since $\oa^0=\ot$,
  i.e., all vertices are thinking. 

  For the inductive case, we wish to show that
  $\id{safe}(E,\oa)$ implies $\id{safe}(E,\oa')$, where
  \[\oa' = f_P(\oa, \ob, \oc)\]

  Let $E(j,k)$.  Let $\oa'(j)=\e$.  In each case, we prove
  that $\oa'(k)\neq \e$.

\be

  \item $\oa(j)=\t$:  By Lemma~\labelcref{lem:t-next},
    $\oa'(j)\neq \e$.   This violates the assumption that $\oa'(j)=\e$.

  \item $\oa(j)=\e$:  There  are three cases:
   \be

    \item\label{itm:kt} $\oa(k)=\t$: Again, by
      Lemma~\labelcref{lem:t-next}, $\oa'(k)\neq\e$.

    \item $\oa(k)=\h$: $\oc=g_H(D,\oa)$.  $k$ is hungry and it has an eating
       neighbour $j$.   Hence $\oc(k)=\overzero$.  Now
       \begin{align*}
         \oa'(k) &=  f_S(\oa(k), \ob(k), \oc(k))\\
         \oa'(k) &=  f_S(\oa(k), \ob(k), \overzero)\\
         \oa'(k) &=  f_P(\oa(k), 0)\\
                 &=  \oa(k)\\
                 &=  \h
       \end{align*}
       Thus $\oa'(k)\neq\e$.

      \item  $\oa(k)=\e$:   This is ruled out by the
       induction hypothesis because $\oa$ is safe.
    \ee

  \item $\oa(j)=\h$:  Again, there are three cases:
      \be
        \item $\oa(k)=\t$:  $\oa'(k)\neq\e$ follows from
          Lemma~\labelcref{lem:t-next}.  

        \item $\oa(k)=\h$:  There are two cases:
          \be
              \item $j\mapsto k\in D$:  From
                Lemma~\labelcref{lem:D-dom}, $\oa'(k)\neq \e$.
              \item $k\mapsto j\in D$:  By identical
                reasoning, $\oa'(j)\neq \e$, which
                contradicts the assumption that
                $\oa'(j)=\e$. 
          \ee
        \item $\oa(k)=\e$: By the induction hypothesis
          (eaters are sinks) applied to $D$ and the fact
          that $E(j,k)$, it follows that $j\mapsto k\in D$.
          From Lemma~\labelcref{lem:h-if-n-eats}, it follows that
          $\oa'(j)\neq\e$, which contradicts the assumption
          that $\oa'(j)=\e$.  \ee \ee

\ee

\end{proof}

\subsection{Starvation freedom}
\label{subsec:starvation-freedom}

Starvation-freedom means that every hungry vertex eventually eats.
The argument for starvation freedom is built over the several lemmas.

The first of these asserts a central property of the
priority map, that it is acyclic.
\begin{lemma}[Priority Map is acyclic]
\label{lem:D-acyclic}
$D$ is acyclic.
\end{lemma}

\begin{proof}

The proof is by induction on $i$.

  The base case is true because $D^{0}$ is acyclic
  by construction. 

  For the inductive case, we need to show that $D'$ is
  acyclic.  Assume, for the sake of deriving a
  contradiction, that $D'$ has a cycle.  Since
  $D'=d_H(D,\oa')$ and $D$ is acyclic by the induction
  hypothesis, the cycle in $D'$ must involve an edge $d'$ in
  $D'$ but not in $D$.

  $d'$ is an edge $k\mapsto j$.   There are two possibilities
  based on the first two clauses of the definition of $d_H$:

   \be

    \item $\oa'(j)=\e$: In that case by
      \cref{thm:safety-invariants}~\labelcref{itm:e-sink},
        $j$ is a sink in $D'$.  If $d'$ is part of a cycle
        in $D'$, then $j$ is part of that cycle, but since
        $j$ is a sink, it can not participate in any cycle.
        Contradiction.

    \item $\oa'(j)=\t$ and $\oa'(k)=\h$: Since $k\mapsto j$
      is an edge in $D'$, there is a path 
       \[j\rightarrow l_1\mapsto\ldots l_m\mapsto l_{m+1}\ldots \mapsto k\]
       in $D'$. 

       Then it must be the case that for some $m$,
       $\oa'(l_m)=\t$ and $\oa'(l_{m+1})=\h$ and
       $D'(l_m,l_{m+1})$.  But by
       \cref{thm:safety-invariants}~\labelcref{itm:h-dom-t},
       $D'(l_{m+1}, l_m)$.  We can not have both
      $D'(l_m,l_{m+1})$ and $D'(l_{m+1}, l_m)$.
      Contradiction.
   \ee

\end{proof}

The next set of lemmas demonstrate how the function $d_H$
transforms the subordinate and dominator set of a hungry
vertex that is left unchanged in the next step. 

\begin{lemma}[Monotonicity of subordinate set and Anti-monotonicity of dominator set]
\label{lem:shrink-grow}

Let $\oa(j)=\h$ and $\oa'(j)=\h$, 
\be
  \item\label{itm:sub-doesnot-shrink} \textbf{Subordinate
    set monotonicity:} $D(j)\subseteq D'(j)$.

  \item\label{itm:dom-doesnot-grow} \textbf{Dominator set
    anti-monotonicity:} $D'^{-1}(j)\subseteq
    D^{-1}(j)$.  \ee
\end{lemma}

\begin{proof}
The proof relies on examining the clauses of the definition
$d_H$:

\be
   \item \textbf{Subordinate set monotonicity:} Let
     $D(j,k)$.  We wish to prove that $D'(j,k)$.  

     Now, $D'=d_H(D,\oa')$.   Consider $d_H(j\mapsto k,
     \oa')$:   Since $\oa'(j)=\h$, it follows, from the
     third clause of the definition of $d_H$ that $j\mapsto
     k\in D'$, i.e., $D'(j,k)$. 

  \item  \textbf{Dominator set anti-monotonicity:}  Let
    $D'^{-1}(j,k)$.  We wish to show that $D^{-1}(j,k)$.
    $D'^{-1}(j,k)$ means that $D'(k,j)$.   Similarly,
    $D^{-1}(j,k)$ means that $D(k,j)$.     

    Thus we are given that $D'(k,j)$ and we need to show
    that $D(k,j)$.  Now, $D'=d_H(D,\oa')$ and
    $\oa'(k)=\h=\oa'(j)$.  We reason backwards with the
    definition of $d_H$.  We are given something in the
    range $D'$ of $d_H$, we reason why it also exists in the
    domain $D$.  In the definition of $d_H$, only the last
    clause is applicable, which leaves the edge unchanged.
    Since $D'(k,j)$, it follows that $D(k,j)$.
\ee

\end{proof}

As a corollary, a hungry vertex that is top and continues to
be hungry in the next step also continues to be a top
vertex.

\begin{corollary}[Top continues]
\label{cor:top-continues}
  Given that $\oa(j)=\h=\oa'(j)$, and $\id{top}(D)(j)$, it
  follows that $\id{top}(D')(j)$.
\end{corollary}
\begin{proof}
From \cref{lem:shrink-grow},
part~\labelcref{itm:dom-doesnot-grow}, 
$D'^{-1}(j)\subseteq D^{-1}(j)$.   Since $\id{top}(D)(j)$, it means that
$D^{-1}(j)=\emptyset$.   It follows
that $D'^{-1}(j)=\emptyset$, i.e., $\id{top}(D')(j)$.  
\end{proof}

The next lemma examines the set of eating neighbours of a
top hungry vertex after a step that leaves the vertex hungry
and at the top.  

\begin{lemma}[No new eating neighbours if top continues]
\label{lem:no-new-eating-nbrs}

Let $\oa(j)=\h=\oa'(j)$ and $\id{top}(D)(j)$.   Then
$D'_{\e}(j)\subseteq D_{\e}(j)$.  
\end{lemma}

\begin{proof}
For the sake of  deriving a contradiction, assume
that $D'_{\e}(j)\not\subseteq D_{\e}(j)$.  

Then there is some vertex $k$ such that $k\in D'_{\e}(j)$
and $k\not\in D_{\e}(j)$.  Since $k\in D'_{\e}(j)$,  $k$ is
a neighbour of $j$.  We are given that $\oa'(k)=\e$ and
$\oa(k)\neq\e$.  

This leaves us with two possibilities: 

 \be
    \item $\oa(k)=\t$:  Then, by \cref{lem:t-next}, $k$'s
      activity cannot be $\e$ in the next step.  Therefore
      $k\not\in E'_{\e}(j)$.   Contradiction.

    \item $\oa(k)=\h$: Since $j$ is top in $D$, $j\mapsto
      k\in D$.  Then by \cref{lem:D-dom}, $\oa'(k)=\h$, so
      $k\not\in E'_{\e}(j)$.  Contradiction again.  \ee

\end{proof}

The next lemma generalises the second part of the previous
lemma and relates the closure of the dominator set of a
hungry vertex going from one step to the next.

\begin{lemma}[Transitive closure of the dominator set does
    not grow]
\label{lem:transitive-dominator}

If $\oa(j)=\h=\oa'(j)$, then $D'^{-1+}(j)\subseteq
D^{-1+}(j)$. 
\end{lemma}

\begin{proof}
If $\oa(j)=\h$, let $P(j)$ denotes the length of the longest
path from a top vertex in $D$ to $j$.  Note that $P$ is well
defined since $D$ is acyclic.  Also, if $k\mapsto j\in D$,
then, from
\cref{thm:safety-invariants},~parts~\labelcref{itm:e-sink}
(Eaters are sinks) and~\labelcref{itm:h-dom-t} (Hungry
dominate thinkers), $\oa(k)=\h$ and therefore $P(k)$ is
well-defined, and furthermore, $P(k) < P(j)$.

The proof is by induction on $P(j)$.

\textbf{Base case: } $P(j)=0$.  This implies that $j$ is a
top vertex in $D$ and therefore $D^{-1}(j)=\emptyset$.  Then
from \cref{cor:top-continues}, $j$ is a top vertex in $D'$
that is hungry.  Thus $D^{-1}(j)=\emptyset$ and the result
follows.

\textbf{Inductive case: } $P(j)>0$.
Now 

\[D^{-1+}(j) = D^{-1}(j)\union \set{D^{-1+}(k)\ |\ k
  \in D^{-1}(j)}\]

and, similarly

\[D'^{-1+}(j) = D'^{-1}(j)\union \set{D'^{-1+}(k)\ |\ k \in D'^{-1}(j)}\]

Let $l\in D'^{-1+}(j)$.   There are two cases:

\be
\item $l\in D'^{-1}(j)$:  From \cref{lem:shrink-grow},
  part~\labelcref{itm:dom-doesnot-grow}, it follows that
  $l\in D^{-1}(j)$ and hence $l\in D^{-1+}(j)$.  

\item $l\in D'^{-1+}(k)$ for some $k\in D'^{-1}(j)$: Then,
  by another application of \cref{lem:shrink-grow},
  part~\labelcref{itm:dom-doesnot-grow}, it follows that
  $k\in D^{-1}(j)$.  That is $k\mapsto j\in D$.  That means
  that $P(k)<P(j)$.  

  Applying the induction hypothesis on $k$, we have
  $D'^{-1+}(k)\subseteq D^{-1+}(k)$.  Hence $l\in
  D^{-1+}(k)$.  From this it follows that $l\in
  D^{-1+}(j)$.  
\ee

\end{proof}

We now prove that the N Diners with centralised controller
exhibits starvation freedom.  
\begin{theorem}[starvation freedom]
\label{thm:starvation-freedom}

The system of N Dining Philosophers meets the following
starvation freedom properties:
\begin{enumerate}

\item\label{itm:eat-finishes} \textbf{Eater eventually
  finishes:} If $j$ is eating, then $j$ will eventually
  finish eating.

\item\label{itm:top-eats} \textbf{Top eventually  eats:}  If $j$ is a
  hungry top vertex, then $j$ will eventually start
  eating.

\item\label{itm:hungry-tops} \textbf{Hungry eventually
  tops:} If $j$ is a hungry vertex that is not top, then $j$
  will eventually become a top vertex.
\end{enumerate}

From the above three properties, one may conclude that a
hungry vertex eventually eats.  
\end{theorem}

The proof of this theorem hinges on defining an appropriate
set of metrics on each behaviour of the Dining Philosophers
problem.  

\begin{proof}

We define a set of metrics that map a hungry or eating
vertex $j$ to a natural number:

\be
  \item Let $W_{\e}:V\rightarrow \N\rightarrow \N$ be
    defined as follows:

    $W_{\e}(j)[i]=0$ if $\oa(j)\neq\e$,
    otherwise $W_{\e}(j)[i]$ is equal to the number of steps
    remaining before $j$ finishes eating.  Clearly,
    $W_{\e}(j)$ is positive as long as $j$ eats and $0$
    otherwise.

  \item Let $W_{\tp}: V\rightarrow \N\rightarrow \N$ be defined as follows:

   \begin{align*}
      W_{\tp}(j)[i] &= \Sigma_{k\in D_{\e}[i](j)} W_{\e}(k)[i], \quad  \mbox{if $\id{top}(D,j)$ and $\oa(j)=\h$}
   \end{align*}
     Note that $W_{\tp}(j)[i]$ is positive as long as $j$ is
     a hungry top vertex that is not ready, and $0$
     otherwise. 
  \item If $\oa(j)=\h$:

 \begin{align*}
   W_{\h}(j)&\ida [|D^{-1+}(j)|, \\
           &  \quad \quad \Sigma\set{W_{\tp}(k)\ |\ k\in  D^{-1+}(j)\ \land \id{top}(D)(k)}]
 \end{align*}
\ee

$W_{\h}$ is a pair $[v_1,v_2]$.  $W_{\h}$ is well defined
since $D$ is acyclic.  The ordering is lexicographic:
$[v'_1,v'_2] < [v_1, v_2]$ iff $v'_1 < v_1$ or $v'_1=v_1$
and $v'_2 < v_2$.  $D^{-1+}(j)$ denotes the transitive
closure of $\set{j}$ with respect to $D^{-1}$.  If $j$ is is
not at the top, then $D^{-1+}(j)\neq \emptyset$ and
therefore $v_1>0$.  If $j$ is at the top then
$W_{\h}(j)=[0,0]$.

We prove the following:

\be

 \item $W_{\e}$ is a decreasing function:  If $\oa(j)=\e
   =\oa'(j)$, then $W'_{\e}(j) < W_{\e}(j)$.  The proof is
   obvious from the definition of $W_{\e}$. 

 \item\label{itm:Wtop-dec} $W_{\tp}$ is a decreasing
   function: If $\oa(j)= \h = \oa'(j)$, then $W_{\tp}'(j) <
   W_{\tp}(j)$.  Since $j$ is a hungry top vertex in $D$ and
   hungry in $D'$, it follows from \cref{cor:top-continues}
   (Top continues) that $j$ is a top vertex in $D'$.

\begin{align*}
W'_{\tp}(j) &= \Sigma_{k\in D'_{\e}(j)} W_{\e}'(k)\\
           &< \Sigma_{k\in D_{\e}(j)} W_{\e}(k)\\
           &= W_{\tp}(j)
\end{align*}

The penultimate inequality holds because of the following
two reasons: From
\cref{lem:no-new-eating-nbrs},
 $D'_{\e}(j)\subseteq D_{\e}(j)$.  Second, from the definition of
$W_{\e}$, for each $k\in D'_{\e}$, $W'_{\e}(k)< W_{\e}(k)$.   

The last step holds because $j$ is top in $D$.

  \item\label{itm:Wh-decr} $W_{\h}$ is a decreasing
    function: If $\oa(j)=\h=\oa'(j)$, $j$ is not top in $D$,
    then $W'_{\h}(j)<W_{\h}(j)$.

   To prove this, consider the definition of $W_{\h}(j)$:

       \begin{align*}
           W'_{\h}(j)&= [|D'^{-1+}(j)|, \\
                    &  \quad \quad \Sigma\set{W'_{\tp}(k)\ |\ k\in D'^{-1+}(j)\ \land \id{top}(D')(k)}]
       \end{align*}

From \cref{lem:transitive-dominator}, $D'^{-1+}(j)\subseteq
D^{-1+}(j)$.  There are two cases:

   \be
     \item $D'^{-1+}(j)\subset D^{-1+}(j)$:   Clearly,
       $W'_{\h}(j) < W_{\h}(j)$.

     \item $D'^{-1+}(j) = D^{-1+}(j)$:   Again, there are
       two cases:

        \be 
          \item $D^{-1+}(j)=\emptyset$: then $j$ is a hungry
            top vertex $D$.  This violates the assumption
            that $j$ is not top in $D$. 

          \item $D^{-1+}(j)\neq\emptyset$: Then, for each
            top vertex $k$ in $D^{-1+}(j)$ and
            $D'^{-1+}(j)$, $k$ is in the domain of $W_{\tp}$
            and $W'_{\tp}$.  Furthermore, from
            part~\labelcref{itm:Wtop-dec}, $W'_{\tp}(k)<
            W_{\tp}(k)$.  The result follows:
            $W'_{\h}(j)<W_{\h}(j)$.

\ee \ee





\ee

\end{proof}

\section{Distributed Solution to the N Diners problem}
\label{sec:dist}

In the distributed version of N Diners, each philosopher
continues to be connected to other philosophers adjacent to
it according to $E$, but there is no central hub controller.
Usually the problem is stated as trying to devise a protocol
amongst the philosophers that ensures that the safety and
starvation freedom conditions are met.  The notion of devising a
protocol is best interpreted as designing a collection of
systems and their composition.

\subsection{Architecture and key idea}
The centralised architecture employed the global maps $\oa$,
$\ob$, $\oc$ and $D$.  While the first three map a vertex
$j$ to a value (activity, choice input, or control) the last
maps an edge $\set{j,k}$ to one of the vertices $j$ or $k$.

The key to devising a solution for the distributed case is
to start with the graph $G=\pair{V,E}$ and consider its
distributed representation.  The edge relation $E$ is now
distributed across the vertex set $V$.  Let $\alpha_j$
denote the size of the set of neighbours $E(j)$ of $j$.  We
assume that the neighbourhood $E(j)$ is arbitrarily ordered
as a vector $\vec{E_j}$ indexed from $1$ to $\alpha_j$.  Let
$j$ and $k$ be distinct vertices in $V$ and let $\set{j,k}\in
E$.  Furthermore, let the neighbourhoods of $j$ and $k$ be
ordered such that $k$ is the $m$th neighbour of $j$ and $j$
is the $n$th neighbour of $k$.  Then, by definition,
$\vec{E_j}(m)=k$ and $\vec{E_k}(n)=j$.

In addition, with each vertex $j$ is associated a
philosopher system $S_j$ and a \emph{local} controller
system $L_j$.  The philosopher system $S_j$ is an instance
of the system $S$ defined in \cref{subsec:S}.  In designing
the local controllers, the guiding principle is to
distribute the state of the centralised controller to $N$
local controllers.  The state of the centralised controller
consists of the directed graph $D$ that maps each edge in
$E$ to its dominating endpoint and the map
$\oc:V\longrightarrow \id{Cmd}$ which is also the output of
the hub controller.

The information about the direction of an edge $\set{j,k}$
is distributed across two \emph{dominance vectors}
$\vec{d_j}$ and $\vec{d_k}$.  Both are boolean vectors
indexed from 1 to $\alpha_j$ and $\alpha_k$, respectively.
Assume that $k=\vec{E_j}(m)$ and $j=\vec{E_k}(n)$.  Then,
the value of $D(\set{j,k})$ is encoded in $\vec{d_j}$ and
$\vec{d_k}$ as follows: If $D(\set{j,k})=j$ then
$\vec{d_j}(m)=\true$ and $\vec{d_k}(n)=\false$.  If
$D(\set{j,k})=k$, then $\vec{d_j}(m)=\false$ and
$\vec{d_k}(n)=\true$.

In the next subsection we define the local controller as a
Tabuada system.

\subsection{Local controller system for a vertex $j$}
\label{subsec:local-controller-defn}

The controller system $L_j$ has $\alpha_j+1$ input ports of
type $A$ which are indexed $0$ to $\alpha_j$.  The output of
$L_j$ is of type $\id{Cmd}$.

The local controller $L_j$ is a Tabuada system
\[L_j = \pair{X,X^{0},U,f, Y, h}\]

\noindent where 

\be
 \item $X=([1..\alpha_j]\longrightarrow \Bool)\times
   \id{Cmd}$.  Each element of $X$ is a tuple
   $(\vec{d_j},c_j)$ consisting of a \emph{dominance vector}
   $\vec{d_j}$ indexed $1$ to $\alpha_j$ and a command value
   $c_j$.  $\vec{d_j}(m)=\true$ means that there is a
   directed edge from $j$ to its $m$th neighbour $k$;
   $\false$ means that there is an edge from its $m$th
   neighbour to $j$. 

 \item $X^{0}$ is defined as follows:
   $X^0=\pair{\vec{d^{0}_j},c^{0}_j}$ where $c^{0}_j=\pass$
   and $\vec{d^{0}_j}(m) = \true$ if $\vec{E_j}(m)=k$
   and $j>k$, $\false$ otherwise.   In other words, there is
   an edge from $j$ to $k$ if $j>k$. 

 \item $U=[0..\alpha_j]\longrightarrow A$: We denote the
   input to $L_j$ as a vector $\vec{a_j}$, the activities of
   all the neighbours of the $j^{th}$ philosopher, including
   its own activity.  $\vec{a_j}(m)$ denotes the value of
   the $m$th input port.

 \item $f_L:X, U\longrightarrow X$ defines the dynamics of
   the controller and is given below.

 \item $Y=\id{Cmd}$, and

 \item $h:X\rightarrow Y$ and $h(\vec{d_j},c_j)=c_j$.
   The output of the controller $L_j$ is denoted $c_j$.  
\ee

The function $f_L$ takes a dominance vector $\vec{d}$ of
length $M$, a command $c$ and an activity vector $\vec{a}$
of length $M+1$ and returns a pair consisting of a new
dominance vector $\vec{d'}$ of length $M$ and a new command
$c'$.  $f_L$ first computes the new dominance vector
$\vec{d'}$ using the function $d_L$.  The result $\vec{d'}$
is then passed along with $\vec{a}$ to the function $g_L$,
which computes the new command value $c'$.  The functions
$f_L$ and $d_L$ are defined below:

\[f_L((\vec{d},c),\vec{a}) = (\vec{d'}, c') \quad \mbox{where}\]

\begin{align}
  \vec{d'} &= \vec{d_L}(\vec{d},\vec{a}),\ \mbox{and}\\
  c'       &= g_L(\vec{d'},\vec{a})
\end{align}

\begin{align}
\vec{d_L}(\vec{d},\vec{a})(m) &\ida
d_L(\vec{d}(m),\vec{a}(0),\vec{a}(m)) \quad \mbox{where
  $m\in [1..M]$}
\end{align}

$d_L(d,a_0,a)$ is  defined as
\begin{align}
d_L(d, \t, \t) &=d\label{eqn:dLdef-tt}\\
d_L(d, \t, \h) &=\false\label{eqn:dLdef-th}\\
d_L(d, \t, \e) &=\true\label{eqn:dLdef-te}\\
d_L(d, \h, \e) &=\true\label{eqn:dLdef-he}\\
d_L(d, \h, \h) &=d\label{eqn:dLdef-hh}\\
d_L(d, \e, \h) &=\false\label{eqn:dLdef-eh}\\
d_L(d, \e, \t) &=\false\label{eqn:dLdef-et}\\
d_L(d, \h, \t) &=\true\label{eqn:dLdef-ht}\\
d_L(d, \e, \e) &=d\label{eqn:dLdef-ee}
\end{align}

$d_L(\vec{d}(m), \vec{a}(0), \vec{a}(m))$ takes the $m$th
component of a dominance vector $\vec{d}$ and computes the
new value based on the activity values at the $0$th and
$m$th input ports of the controller.


The function $g_L$ takes a dominance vector $\vec{d}$ of
size $M$ and an activity vector $\vec{a}$ of size $M+1$ and
computes a command.  It is defined as follows:
\begin{align*}
  g_L(\vec{d},\vec{a}) &\ida \pass, \quad \quad \mbox{if $\vec{a}(0)\in\set{\t,\e}$}\\
                      &\ida \overone, \quad \mbox{if $\id{ready}_L(\vec{d},\vec{a})=\true$}\\
                      &\ida \overzero, \quad \mbox{otherwise}\\\\\nonumber
  \id{ready}_L(\vec{d}, \vec{a}) &\ida \true, \quad \mbox{if  $\vec{a}(0)=\h$ and $\id{top}_L(\vec{d})$ and %
                            $\forall m\in [1..M]:\ \vec{a}(m)\neq\e$}\\
                           &\ida \false,  \quad \mbox{otherwise}\\\\\nonumber
  \id{top}_L(\vec{d}) &\ida \true,  \quad \mbox{if $\forall m\in [1..M]: \vec{d}(m)=\true$}\\
                      &\ida \false,  \quad \mbox{otherwise}
\end{align*}

Now we can write down the equations that define the
asynchronous dynamics of the philosopher system.  Consider
any arbitrary philosopher $j$ and its local controller $L_j$:
\begin{align}
  a_j^0          &= \t\\
  \mbox{For $m \in [1..\alpha_j]:$} \quad
  \vec{d_j}^0(m) &= \true, \quad \mbox{if $\vec{E_j}(m)=k$
    and $j>k$}\\\nonumber
  &= \false, \quad \mbox{otherwise}\\\nonumber\\
  c_j         &= g_L(\vec{d_j},\vec{a_j})\\
  a_j'        &= f_S(a_j, b_j, c_j)\label{eqn:dist-a}\\
  \vec{d_j}'  &= d_L(\vec{d_j}, \vec{a_j'})
\end{align}

Note from equation \labelcref{eqn:dist-a} that the philosopher dynamics has not changed - it is
the same as that of the centralised case.  A close
examination of the equations help us deduce that the
dynamics we obtained in the distributed case are very much
comparable to that of the centralised case.  This identical
nature of the dynamics form the foundation for the
correctness proofs which follow later.

\subsection{Wiring the local controllers and the philosophers}
\label{subsec:dist-wiring}

Each philosopher $S_j$ is defined as the instance of the
system $S$ defined in \cref{subsec:S}.  Let the choice
input, control input and output of the philosopher system
$S_j$ be denoted by the variables $S_j.c$, $S_j.\bbot$ and
$S_j.a$, respectively.  The output of $L_j$ is fed as the
control input to $S_j$.  The output $S_j$ is fed as $0$th
input of $L_j$.  In addition, for each vertex $j$, if $k$ is
the $m$th neighbour of $j$, i.e., $k=\vec{E_j}(m)$, then the
output of $S_k$ is fed as the $m$th input to $L_j$.  (See
\cref{fig:dist-n-diners-arch}).

\begin{figure}
\caption{Wiring between the systems of adjacent philosophers
  $j$ and $k$  where $k$ and
$j$ are respectively the $m$th and $n$th neighbour of each
other\label{fig:dist-n-diners-arch}.}
\centering{
 \includegraphics[width=5in]{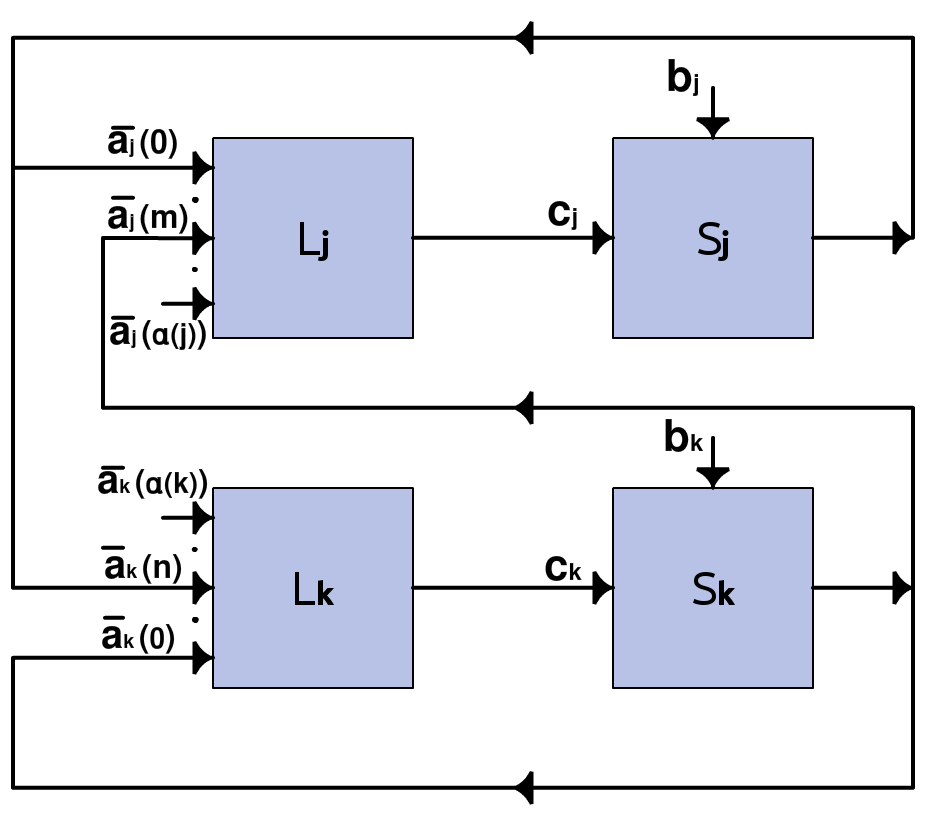}
}
\end{figure}

The wiring between the $N$ philosopher systems and the $N$
local controllers is the interconnect relation $\I\subseteq
\Pi_j S_j.X \times S_j.U\times L_j.X \times L_j.U$, $1\leq
j\leq N$ defined via the following set of constraints:

\be

\item $c_j=S_j.c$: The output of the local controller $L_j$
  is equal to the control input of the philosopher system
  $S_j$.

\item $S_j.a = \vec{a_j}(0)$: the output of the
  philosopher $S_j$ is fed back as the input of the 0th
  input port of the local controller $L_j$.

\item $S_k.a = \vec{a_j}(m)$, where $1\leq m\leq \alpha_j$
  and $k=\vec{E_j}(m)$: the output of the philosopher $S_k$
  is connected as the input of the $m$th input port of the
  local controller $L_j$ where $k$ is the $m$th neighbour of
  $j$.

\item $\vec{d_j}(m)=\neg \vec{d_k}(n)$, where
  $k=\vec{E_j}(m)$ and $j=\vec{E_k}(n)$.  The dominance
  vector at $j$ is compatible with the dominance vectors of
  the neighbours of $j$.
\ee

\subsection{Correctness of the solution to the Distributed case}
\label{subsec:correctness-dist}

The correctness of the solution for the distributed case
rests on the claim that under the same input sequence, the
controllers and the philosopher outputs in the distributed
and centralised cases are identical.  This claim in turn
depends on the fact that the centralised state may be
reconstructed from the distributed state.  

\begin{theorem}[Correctness of Distributed Solution to N
    Diners]
\label{thm:dist-equals-central}

Consider the sequence of lifted choice inputs fed to both
the centralised and the distributed instances of an N Diners
problem.  We show that after equal number of computations,
for each $j\in V$:

 \be
  \item $\oa(j) = a_j$: The output $\oa(j)$ of the $j$th
    Philosopher in the centralised architecture is identical
    to the output $a_j$ of the $j$th Philosopher in the
    distributed architecture.

  \item $\oc(j) = c_j$: $\oc(j)$, the $j$th output of hub
    controller in centralised architecture is identical to
    the output $c_j$ of the $j$th local controller in the
    distributed architecture.

  \item For each $k\in E(j)$, 
      \be 
        \item $k=\vec{E_j}(m)$ for some $m\in [1..\alpha_j]$, and
        \item $j=\vec{E_k}(n)$ for some $n\in [1..\alpha_k]$, and
        \item $j\mapsto k\in D$ iff $\vec{d_j}(m)=\true$
          and $\vec{d_k}(n)=\false$.
      \ee
 \ee 
\end{theorem}
\begin{proof}

Note that in the last clause it is enough to prove one
direction (only if) since, if $k\in E(j)$ and $j\mapsto
k\not\in D$ implies $k\mapsto j\in D$; the proof of this
case is simply an instantiation of the theorem with $k$
instead of $j$. 

The proof for the rest of the conditions is by induction on
the number of computations\footnote{If $k$ is the value of a
  variable after $i$ computations, then $k'$ stands for
  value after $i+1$ computations.  All non-primed variables
  are assumed to have undergone the same number of
  computations.  Same is the case with primed variables, but
  with one extra computation than its non-primed version.}.

For the base case, in the centralised architecture, the
initial values for each vertex $j$ are $\oa^{0}(j)=\t$,
$\oc^{0}(j)=\pass$, and for each $k\in E(j)$, $D^{0}(\set{j,k})$ is
equal to $j$ if $j>k$, and $k$ otherwise.

In the distributed regime, the initial value $a_j$ of the
output of Philosopher $S_j$ is $\t$ by definition of $S_j$.
The initial value of the output $c_j$ of the local
controller $L_j$ is $\pass$ by definition.  Also, note that
in the initial state $\vec{d_j}(m)=\true$ iff
$k=\vec{E_j}(m)$ and $j>k$, $\false$ otherwise.

Assume, for the sake of the induction hypothesis, the
premises above are all true.  We wish to show that

 \be
  \item $\oa'(j) = a'_j$ 
  \item $\oc'(j) = c'_j$ 
  \item For each $k\in D'(j)$, 
      \be 
        \item $k=\vec{E_j}(m)$ for some $m\in [1..\alpha_j]$, and
        \item $j=\vec{E_k}(n)$ for some $n\in [1..\alpha_k]$, and
        \item $j\mapsto k\in D'$ iff $\vec{d'_j}(m)=\true$
          and $\vec{d'_k}(n)=\false$.
      \ee 
 \ee 

We start with $D'$.  Suppose $j\mapsto k\in D'$.  We need to
show that $\vec{d'_j}(m)=\true$ and $\vec{d'_k}(n)=\false$.
Based on the definition of $d_H$, there are three cases:
\be
  \item Case 1  (clause \labelcref{eqn:dh1} of
    $d_H$):
     \begin{align}
       \oa(k) =\e\\
       k\mapsto j\in D  \label{eqn:k2jinD}
     \end{align} 

Note the substitution requires swapping $k$ and $j$ in
clause \labelcref{eqn:dh1}.

From the above two conditions, applying the inductive
hypthothesis, we have

    \begin{align}
      \vec{d_k}(n) &=\true\\
      \vec{d_j}(m) &=\false\\
      a_k &=\e
    \end{align}

$\vec{d'_j}(m)$ may  now be computed as follows:
    \begin{align*}
      \vec{d'_j}(m) &=  d_L(\vec{d_j}(m), \vec{a_j}(0), \vec{a_j}(m))\\
                    &=  d_L(\false, a_j, a_k) \\
                    &=  d_L(\false, a_j, \e)
    \end{align*}

An inspection of the definition of $d_L$ reveals three cases
that have $\e$ in the third argument:
\labelcref{eqn:dLdef-te}, \labelcref{eqn:dLdef-he}, and
\labelcref{eqn:dLdef-ee}.  Of these, the last case is ruled
out because of the safety property of the centralised
solution; no two adjacent vertices eat at the same time.
For each of the other two cases, $\vec{d'_j}(m)$ yields
$\true$.

Computing $\vec{d'_k}(n)$, 
    \begin{align*}
      \vec{d'_k}(n) &=  d_L(\vec{d_k}(n), \vec{a_k}(0), \vec{a_k}(m))\\
                    &=  d_L(\true, a_k, a_j) \\
                    &=  d_L(\true, \e, a_j)
    \end{align*}

A similar examination of cases \labelcref{eqn:dLdef-et},
\labelcref{eqn:dLdef-eh} allows us to conclude that
$\vec{d'_k}(n)=\false$.

\item Case 2 (clause \labelcref{eqn:dh2} of $d_H$):
  \begin{align}
   \oa(k)=\t\\
   \oa(j)=\h\\
   k\mapsto j\in D
  \end{align}

Note again that the substitution requires swapping $k$ and
$j$, this time in clause \labelcref{eqn:dh2}.

From the above conditions, applying the induction
hypothesis,
 \begin{align}
   \vec{d_k}(n)=\true\\
   \vec{d_j}(m)=\false\\
   a_k=\t\\
   a_j=\h
 \end{align}

$\vec{d'_j}(m)$ may  now be computed as follows:
    \begin{align*}
      \vec{d'_j}(m) &=  d_L(\vec{d_j}(m), \vec{a_j}(0), \vec{a_j}(m))\\
                    &=  d_L(\false, a_j, a_k) \\
                    &=  d_L(\false, \h, \t)\\
                    &=  \true
    \end{align*}

Computing $\vec{d'_k}(n)$, 
    \begin{align*}
      \vec{d'_k}(n) &=  d_L(\vec{d_k}(n), \vec{a_k}(0), \vec{a_k}(m))\\
                    &=  d_L(\true, a_k, a_j) \\
                    &=  d_L(\true, \t,  \h)\\
                    &=  \false
    \end{align*}

\item Case 3:  From clause \labelcref{eqn:dh3} of $d_H$, it
  follows that

\begin{align}
 j\mapsto k\in D\\
 \oa(j)\neq \e\\ 
 \neg (\oa(j)=\t\ \mbox{and}\ \oa(k)=\h)
\end{align}

By the induction hypothesis
 \begin{align}
   \vec{d_j}(m)=\true\\
   \vec{d_k}(n)=\false\\
   a_j\neq \e\\ 
 \neg (a_j=\t\ \mbox{and}\ a_k=\h)
 \end{align} 

$\vec{d'_j}(m)$ may now be computed as follows:
    \begin{align*}
      \vec{d'_j}(m) &=  d_L(\vec{d_j}(m), \vec{a_j}(0), \vec{a_j}(m))\\
                    &=  d_L(\true, a_j, a_k) \\
    \end{align*}

The conditions on $a_j$ and $a_k$ eliminate the
possibilities \labelcref{eqn:dLdef-th},
\labelcref{eqn:dLdef-eh}, \labelcref{eqn:dLdef-et}, 
and \labelcref{eqn:dLdef-ee} in the definition of $d_L$.  Of the
remaining five cases, three
cases \labelcref{eqn:dLdef-te}, \labelcref{eqn:dLdef-he} and
\labelcref{eqn:dLdef-ht} yield the value $\true$, while the
remaining two cases (\labelcref{eqn:dLdef-tt}
and \labelcref{eqn:dLdef-hh}) yield the value $\vec{d_j}(m)$,
which by induction hypothesis, is also $\true$.

$\vec{d'_k}(n)$ may now be computed as follows:
    \begin{align*}
      \vec{d'_k}(n) &=  d_L(\vec{d_k}(n), \vec{a_k}(0), \vec{a_k}(m))\\
                    &=  d_L(\false, a_k, a_j)\\
    \end{align*}
Again, the condition $a_j\neq \e$ eliminates the
possibilities \labelcref{eqn:dLdef-te},
\labelcref{eqn:dLdef-he}, \labelcref{eqn:dLdef-ee} and
\labelcref{eqn:dLdef-ee} in the definition of $d_L$.  Of the
remaining six cases, the impossibility of the condition
$a_k=\h$ and $a_j=\t$, eliminates the case
\labelcref{eqn:dLdef-ht}.  

Of the remaining five cases, three of them,
\labelcref{eqn:dLdef-th}, \labelcref{eqn:dLdef-eh} and
\labelcref{eqn:dLdef-et} yield the value $\false$, while the
remaining two cases (\labelcref{eqn:dLdef-tt} and
\labelcref{eqn:dLdef-hh}) yield the value $\vec{d_k}(n)$,
which by induction hypothesis, is also $\false$.
\ee

The next thing to prove is the claim $\oc(j)= c_j$ for all
computations.  The proof is by induction: Verify that
$\oc^0(j)=c^0_j$ and $\oc(j)=c_j$ implies $\oc'(j)=c'_j$.

The proof proceeds by first showing that

\begin{align}
j\in \id{top}(D)\     &\mbox{iff}\  \id{top}_L(\vec{d_j})=\true\\
\id{ready}(D,\oa)(j)\ &\mbox{iff }  \id{ready}_L(\vec{d_j},\vec{a_j})=\true\\
g_H(D,\oa)(j)         &=  g_L(\vec{d_j},\vec{a_j})
\end{align}

The proofs of each of these are straightforward and
omitted.   

Finally, in both the centralised and distributed architectures,
the philosopher system instances are identical and hence
they have the same dynamics, they both operate with
identical initial conditions ($\oa(j)=a_j=\t$) and in each
case the choice inputs are identical and the control inputs,
which are outputs of the respective controllers, are
identical as well ($\oc(j) = c_j$ as proved above).  From
this, it follows that $\oa(j)=a_j$ for all computations.

\end{proof}

This concludes our formal analysis of the Generalised N
Diners problem and its solution for centralised and
distributed scenarios.   

\section{Related Work}
\label{sec:related}
This section is in two parts: the first is a detailed
comparison with Chandy and Misra's solution, the second is a
survey of several other approaches.

\subsection{Comparison with Chandy and Misra solution}
\label{subsec:cm-soln}
Chandy and Misra\cite{chandy-misra-book-1988} provides the
original statement and solution to the Generalised Dining
Philosophers problem.  There are several important points of
comparison with their problem formulation and solution.

The first point of comparison is architecture: in brief,
shared variables vs. modular interconnects.  Chandy and
Misra's formulation of the problem identifies the division
between a \emph{user} program, which holds the state of the
philosophers, and the \emph{os}, which runs concurrently
with the user and modifies variables shared with the
\emph{user}.  Our formulation is based on formally defining
the two main entities, the philosopher and the controller,
as formal systems with clearly delineated boundaries and
modular interactions between them.  The idea of feedback
control is explicit in the architecture, not in the shared
variable approach. 

Another advantage of the modular architecture that our
solution affords is apparent when we move from the
centralised solution to the distributed solution.  In both
cases, the definition of the philosopher remains exactly the
same; additional interaction is achieved by wiring a local
controller to each philosopher rather than a central
controller.  We make a reasonable assumption that the output
of a philosopher is readable by its neighbours.  In Chandy
and Misra's solution, the distributed solution relies on
three shared boolean state variables per edge in the
\emph{user}: a boolean variable \emph{fork} that resides
with exactly one of the neighbours, its status \emph{clean}
or \emph{dirty}, and a \emph{request token} that resides with
exactly one neighbour, adding up to $3|E|$ boolean
variables.  These variables are not distributed; they reside
with the \emph{os}, which still assumes the role of a
central controller.  In our solution, the distribution of
philosopher's \emph{and} their control is evident.
Variables are distributed across the vertices: each vertex
$j$ with degree $j$ has $\alpha(j)+1$ input ports of type
$\Act$ that read the neighbours' plus self's activity
status.  In addition, each local controller has, as a
boolean vector $\vec{d}_j$ of length $\alpha(j)$ as part of
its internal state, that keeps information about the
direction of each edge with $j$ as an endpoint.  A pleasant
and useful property of this approach is that the centralised
data structure $D$ may be reconstructed by the union of
local data structures $\vec{d}$ at each vertex.

The second point of comparison is the algorithm and its
impact on reasoning.  Both approaches rely on maintaining
the dominance graph $D$ as a partial order.  As a
result, in both approaches, if $j$ is hungry and has
priority over $k$, then $j$ eats before $k$.  In Chandy and
Misra's algorithm, however, $D$ is updated only when a
hungry vertex transits to eating to ensure that eating
vertices are sinks.  In our solution, $D$ is updated to
satisfy an additional condition that hungry vertices always
dominate thinking vertices.  This ensures two elegant
properties of our algorithm, neither of which are true in
Chandy and Misra: (a) a top vertex is also a maximal element
of the partial order $D$, (b) a hungry vertex that is at the
top remains so until it is ready, after which it starts
eating.  In Chandy and Misra's algorithm, a vertex is at the
top if it dominates only (all of its) hungry neighbours; it
could still be dominated by a thinking neighbour.  It is
possible that a hungry top vertex is no longer at the top if
a neighbouring thinking vertex becomes hungry
(\cref{tbl:top-hungry-non-monotonic}).  This leads us to the
third property that is true in our approach but not in
Chandy and Misra's:  amongst two thinking neighbours $j$ and
$k$, whichever gets hungry first gets to eat first.

\begin{table}
  \caption{Example demonstrating two properties of Chandy
    and Misra's algorithm: (a) a top hungry vertex no longer
    remains top, and (b) In step 3, Vertex 1, which was at
    the top, is hungry, but no longer at the
    top. \label{tbl:top-hungry-non-monotonic}}
\begin{center}
\begin{tabular}{rllll}
\hline
i & G & D & top & remarks\\
\hline
0 & $\set{1:\t, 2:\t, 3:\t}$ & $\set{2\mapsto 1, 3\mapsto 1}$ & $\set{2,3}$ & initial\\
\hline
1 & $\set{1:\h, 2:\t, 3:\h}$ & ditto & $\set{2,3}$ & $3$ at top\\
\hline
2 & $\set{1:\h, 2:\t, 3:\e}$ & $\set{2\mapsto 1, 1\mapsto 3}$ & $\set{1,2}$ & $1$ is at the top\\
\hline
3 & $\set{1:\h, 2:\h, 3:\e}$ & ditto & $\set{2}$ & $2$ is at
the top, not $1$\\
\hline
\end{tabular}
\end{center}
\end{table}

\subsection{Comparison with other related work}




Literature on the Dining Philosophers problem is vast.  Our
very brief survey is slanted towards approaches that ---
explicitly or implicitly --- address the modularity and
control aspects of the problem and its
solution.  \cite{Papatriantafilou97ondistributed} surveys
the effectiveness of different solutions against various
complexity metrics like response time and communication
complexity.  Here, we leave out complexity theoretic
considerations and works that explore probabilistic and many
other variants of the problem.

\subsection{Early works}
Dijkstra's Dining Philosophers problem was formulated for
the five philosophers seated in a circle.  Dijsktra later
generalized it to N philosophers.  Lynch\cite{lynch-1981}
generalised the problem to a graph consisting of an
arbitrary number of philosophers connected via edges
depicting resource sharing constraints.  Lynch also
introduced the notion of an interface description of systems
captured via external behaviour, i.e., execution sequences
of automata.  This idea was popularized by Ramadge and
Wonham\cite{21072} who advocated that behaviour be specified
in terms of language-theoretic properties.  They also
introduce the idea of control to affect behaviour.

Chandy and
Misra\cite{Chandy:1984:DPP:1780.1804,chandy-misra-book-1988}
propose the idea of a dynamic acyclic graph via edge
reversals to solve the problem of fair resolution of
contention, which ensures progress.  This is done by
maintaining an ordering on philosophers contending for a
resource.  The approach's usefulness and generality is
demonstrated by their introduction of the Drinking
Philosophers problem as a generalisation of the Dining
Philosophers problem.  In the Drinking Philosophers problem,
each philosopher is allowed to possess a subset of a set of
resources (drinks) and two adjacent philosophers are allowed
to drink at the same time as long as they drink from
different bottles.  Welch and
Lynch\cite{welch-lynch-1993,lynch-book-1996} present a
modular approach to the Dining and Drinking Philosopher
problems by abstracting the Dining Philosophers system as an
I/O automaton.  Their paper, however, does not invoke the
notion of control.  Rhee\cite{500015} considers a variety of
resource allocation problems, include dining philosophers
with modularity and the ability to use arbitrary resource
allocation algorithms as subroutines as a means to compare
the efficiency of different solutions.  In this approach,
resource managers are attached to each resource, which is
similar in spirit to the local controllers idea.

\subsection{Other approaches}
Sidhu~et~al.\cite{sidhu-pollack-1984} discuss a distributed
solution to a generalised version of the dining philosophers
problem.  By putting additional constraints and modifying
the problem, like the fixed order in which a philosopher can
occupy the forks available to him and the fixed number of
forks he needs to occupy to start eating, they show that the solution is
deadlock free and robust. The deadlock-free condition is
assured by showing that the death of any philosopher
possessing a few forks does not lead to the failure of the
whole network, but instead disables the functioning of only
a finite number of philosophers.  In this paper, the
philosophers require multiple (>2) forks to start eating,
and the whole solution is based on forks and their
constraints. Also, this paper discusses the additional
possibility of the philosophers dying when in possession of
a few forks, which is not there in our paper.

Weidman et al.\cite{Weidman1991} discuss an algorithm
for the distributed dynamic resource allocation problem,
which is based on the solution to the dining philosophers
problem. Their version of the dining philosophers problem is
dynamic in nature, in that the philosophers are allowed to
add and delete themselves from the group of philosophers who
are thinking or eating.  They can also add and delete
resources from their resource requirements. The state space
is modified based on the new actions added: adding/deleting
self, or adding/deleting a resource.  The main difference
from our solution is the extra option available to the
philosophers to add/delete themselves from the group of
philosophers, as well as add/delete the resources available
to them.  The state space available to the philosophers is
also expanded because of those extra options - there are
total 7 states possible now - whereas our solution allows
only 3 possible states (thinking, hungry and eating). Also,
the notion of a 'controller' is absent here - the
philosophers' state changes happen depending on the
neighbours and the resources availability, but there is no
single controller which decides it.

Zhan et al.\cite{Zhan2012} propose a mathematical model for
solving the original version of the dining philosophers
problem by modeling the possession of the chopsticks by the
philosophers as an adjacency matrix.  They talk about the
various states of the matrix which can result in a deadlock,
and a solution is designed in Java using semaphores which is
proven to be deadlock free, and is claimed to be highly
efficient in terms of resource usability.

Awerbuch et al.\cite{Awerbuch1990} propose a deterministic
solution to the dining philosophers problem that is based on
the idea of a "distributed queue", which is used to ensure
the safety property. The collection of philosophers operate
in an asynchronous message-driven environment. They heavily
focus on optimizing the "response time" of the system to
each job (in other words, the philosopher) to make it
polynomial in nature. In our solution, we do not talk about
the response time and instead we focus on the modularity of
the solution, which is not considered in this solution.

A distributed algorithm for the dining philosophers
algorithm has been implemented by Haiyan\cite{Haiyan1999} in
Agda, a proof checker based on Martin-Lof's type theory.  A
precedence graph is maintained in this solution where
directed edges represent precedences between pairs of
potentially conflicting philosophers, which is the same idea
as the priority graph we have in our solution.  But unlike
our solution, they also have chopsticks modelled as part of
the solution in Agda.

Hoover et al.\cite{Hoover1992} describe a fully distributed
self-stabilizing\footnote{Regardless of the initial state,
the algorithm eventually converges to a legal state, and
will therefore remain only in legal states.} solution to the
dining philosophers problem.  An interleaved semantics is
assumed where only one philosopher at a time changes its
state, like the asynchronous dynamics in our solution.  They
use a token based system, where tokens keeps circling the
ring of philosophers, and the possession of a token enables
the philosopher to eat.  The algorithm begins with a
potentially illegal state with multiple tokens, and later
converges to a legal state with just one token.  Our
solution do not have this self-stabilization property, as we
do not have any "illegal" state in our system at any point
of time.

The dining philosophers solution mentioned in the work by
Keane et al.\cite{Keane2001} uses a generic graph model
like the generalized problem: edges between processes which
can conflict in critical section access.  Modification of
arrows between the nodes happens during entry and exit from
the critical section.  They do not focus on aspects like
modularity or equational reasoning, but on solving a new
synchronization problem (called GRASP).

Cargill\cite{Cargill1982} proposes a solution which is
distributed in the sense that synchronization and
communication is limited to the immediate neighbourhood of
each philosopher without a central mechanism, and is robust
in the sense that the failure of a philosopher only affects
its immediate neighbourhood.  Unlike our solution, forks are
modelled as part of their solution.

You et al.\cite{You2010} solve the Distributed Dining
Philosophers problem, which is the same as the Generalized
Dining Philosophers problem, using category theory.  The
phases of philosophers, priority of philosophers,
state-transitions etc. are modelled as different categories
and semantics of the problem are explained.  They also make
use the graph representation of the priorities we have used
in our paper.

Nesterenko et al.\cite{Nesterenko2002} present a solution
to the dining philosophers problem that tolerates malicious
crashes, where the failed process behaves arbitrarily and
ceases all operations.  They talk about the use of
stabilization - which allows the program to recover from an
arbitrary state - and crash failure locality - which ensures
that the crash of a process affects only a finite other
processes - in the optimality of their solution.  

Chang\cite{Chang1980} in his solution tries to decentralise
Dijkstra's solution to the dining philosophers problem by
making use of message passing and access tokens in a
distributed system.  The solution does not use any global
variables, and there is no notion of 'controllers' in the
solution like we have in ours.  Forks are made use of in the
solution.

Datta et al.\cite{Datta2005} considers the mobile
philosophers problem in which a dynamic network exists
where both philosophers and resources can join/leave
the network at any time, and the philosophers can
connect/disconnect to/from any point in the network.  The
philosopher is allowed to move around a ring of resources,
making requests to the resources in the process.  The
solution they propose is self-stabilizing and asynchronous.

\subsection{Supervisory control}
  The idea of using feedback (or supervisory)
control to solve the Dining Philosophers program is not new.
Miremadi~et~al.\cite{Miremadi-et-al-2008} demonstrate how to
automatically synthesise a supervisory controller using
Binary Decision Diagrams.  Their paper uses Hoare
composition but does not describe the synthesised
controller, nor do they attempt to prove why their solution
is correct.  Andova~et~al.\cite{andova-et-al-2012} use the
idea of a central controller delegating part of its control
to local controllers to solve the problem of
self-stabilization: i.e., migrating a deadlock-prone
configuration to one that is deadlock-free using distributed
adaptation.

Similar to our solution, Vaughan\cite{Vaughan1992} presents
centralised and distributed solutions to the dining
philosophers problem.  The centralised solution does not
have a hub controller, but has monitor data structures,
which store information like the number of chopsticks
available to each philosopher, the claims made by a
philosopher on his adjacent chopsticks, etc.  In his
distributed solution, the chopsticks are viewed as static
resources and there are manager processes, like we have
controllers, to control them. But unlike our solution, the
local manager processes only control the chopsticks (with
the help of a distributed queue to sequentialize access to
the chopsticks for the philosophers) and not the
philosophers, and the access to the resources is scheduled
by the philosophers by passing messages between themselves.

Siahaan\cite{Siahaan2015}, in his solution, proposes a
framework containing an active object called 'Table' which
controls the forks and the state transitions of the
philosophers.  The other active objects in the framework are
the philosophers and the timer controller (which issues
timeout instructions to the philosophers to change state).
The table manages the state-change requests of the
philosophers depending on the state of forks, hence serving
a purpose similar to the controllers in our solution.  The
timer object sends instructions to the philosophers for
state change, but our paper does not involve a timer to do
so.

Feedback control has been used to solve other problems too.
Wang et al.\cite{wang2008} model discrete event systems
using Petri nets and synthesise feedback controllers for
them to avoid deadlocks in concurrent software.  Mizoguchi
et al.\cite{mizoguchi2016} design a feedback controller of a
cyber-physical system by composing several abstract systems,
and prove that the controlled system exhibits the desired
behaviour.  Fu et al.\cite{fu2014} model adaptive control for
finite-state transition systems using elements from
grammatical inference and game theory, to produce
controllers that guarantee that a system satisfies its
specifications.

\subsection{Synchronous languages}
Synchronous languages like Esterel, SIGNAL and
Lustre\cite{halbwachs-book-1993} are popular in the embedded
systems domain because synchronicity allows simpler
reasoning with time.  Gamatie\cite{gamatie2009designing}
discusses the N Dining Philosophers problem with the
philosophers seated in a ring.  The example is presented in
the programming language SIGNAL, whose execution model uses
synchronous message passing.  The SIGNAL programming
language also compiles the specifications to C code.  The
solution uses three sets of processes: one for the
philosophers, one for the forks, and one for the main
process used for coordination.  Communication between the
philosophers and the forks happens via signals that are
clocked.  In this respect, the solution is similar to the
one described in this paper.  However, in the solution, each
signal has its own clock (polysynchrony), all derived from a
single master clock.


\section{Conclusion and Future Work}
\label{sec:conc}
This work has three objectives: first, to apply the
idea of feedback control to problems of concurrency; second,
to systematically apply the notion of Tabuada systems and
composition when constructing the problem statement and its
solution, and third, to ensure that the solution is as
modular as possible.  The additional notion that we have had
to rely on is the notion of a global clock for synchronous
dynamics, which has considerably simplified the analysis and
proofs.  In the process, we have also come up with a
different solution, one which reveals how the distributed
solution is a distribution of the state in the centralised
solution.

The solution to Dining Philosophers using this approach
leads us to believe that this is a promising direction to
explore in the future, the formalisation of software
architectures for other sequential and concurrent systems.

\bibliography{ref} 
\bibliographystyle{acm}

\appendix

\end{document}